\documentclass[twoside]{article}

\usepackage[pdftex]{pict2e}
\usepackage{graphicx}
\usepackage{epstopdf}
\usepackage{indentfirst}
\usepackage[numbers]{natbib}
\usepackage{amstext}
\usepackage{enumitem}
\usepackage{stackrel}
\usepackage{mathpazo}
\usepackage{amsfonts}
\usepackage{amssymb}
\usepackage{amsmath}
\usepackage{latexsym}
\usepackage{mathrsfs}
\usepackage{color}
\usepackage{dsfont}
\usepackage{url,stackrel,enumitem}	
\usepackage{hyperref}
\usepackage{amsfonts}
\usepackage{stmaryrd}
\usepackage{euscript}
\usepackage{amscd}
\usepackage{bm}
\usepackage{subfig}

\newtheorem{theorem}{Theorem}[section]
\newtheorem{lemma}{Lemma}[section]
\newtheorem{proposition}{Proposition}[section]

\newtheorem{example}{Example}[section]

\newtheorem{definition}{Definition}[section]

\newtheorem{hyp}{Assumption}[section]
\newcommand{\bhyp }{\begin{hyp} \rm }
\newcommand{\ehyp }{\end{hyp}}

\newcommand\I{\mathds{1}}

\newcommand{\lbn}{-1}
\newcommand{\pp}{\psi}

\newcommand{\wh}{\widehat}
\newcommand{\whc}{\widehat{c}}
\newcommand{\whH}{\widehat{H}}

\newcommand{\ff}{{\mathbb F}}

\newcommand{\Q}{{\mathbb Q}}

\newcommand{\bbP}{{\mathbb P}}
\newcommand{\bbR}{{\mathbb R}}

\newcommand{\cF}{\mathcal F}

\newcommand{\CCF}{\mbox{\rm CF}}

\begin{document}
\title{{\Large \bf  Equity Protection Swaps: A New Type of Investment Insurance \\ for Holders of Superannuation Accounts} \vskip 35 pt }

\author{Huansang Xu$\,^{a}$, Ruyi Liu$\,^{a}$ and Marek Rutkowski$\,^{a,b}\,$   \\ \\ \\ \\
$^{a\,}$School of Mathematics and Statistics, University of Sydney \\ Sydney, NSW 2006, Australia \\ \\
$^{b\,}$Faculty of Mathematics and Information Science, Warsaw University of Technology \\ 00-661 Warszawa, Poland \\ }


\date{\vskip 45 pt \today \vskip 35 pt}
\maketitle


\begin{abstract}
We propose to develop a new class of investment insurance products for holders of superannuation accounts in Australia,
which we tentatively call \textit{equity protection swaps} (EPSs). An EPS is a standalone financial derivative, which is reminiscent of a \textit{total return swap} but also shares some features with the variable annuity known as the \textit{registered index-linked annuity} (RILA). The buyer of an EPS obtains a partial protection against losses on a reference portfolio and, in exchange, agrees to share portfolio gains with the insurance provider if the realised return on a reference portfolio is above a predetermined threshold. Formally, a generic EPS consists of protection and fee legs with participation rates agreed upon by the provider and holder. A general fair pricing formula for an EPS is obtained by considering a static hedging strategy based on traded European options. It is argued that to make the contract appealing to holders, the provider should select an appropriate protection and fee rates that make the fair premium at contract's inception equal to zero.
A numerical study based on the Black-Scholes model and empirical tests based on market data for S\&P~500 and S\&P/ASX~200 indices for 2020-2022 demonstrate the benefits of an EPS as an efficient investment insurance tool for superannuation accounts.

\end{abstract}

\newpage

\section{Introduction}    \label{sec1}

Since the 1990s, we observe the global trend of shifting pension plans from defined benefit (DB) schemes to defined contribution (DC) structures. Within a defined contribution type pension scheme, the ultimate retirement benefit hinges upon the accumulated contributions and the effectiveness of their investment over time. Generally speaking, under a DC plan a proportion of employee's earnings is deposited into an investment fund so that the DC account holder may lawfully access the money at their retirement. In contrast to defined benefit pension plans, DC pension funds provide greater control over retirement savings, albeit introducing market risks to participants. Nowadays, a growing number of DC pension account holders are undergoing transition from the accumulation phase toward pension phase. Given the substantial balances at hand, there is a growing interest in managing the market risks and keeping a proportion of the balance on risky assets with a limited potential loss. For concreteness, we will refer to the pension system in Australia but analogous schemes are available in all major economies.

The Superannuation Guarantee (SG) is a comprehensive retirement savings scheme, which was introduced by the Australian Government on 1 July 1992. Under the SG scheme (also known as {\it superannuation} or, simply, {\it super})  employers are required to contribute mandatory sums to individual superannuation accounts of their employees. As of July 2023, the minimum employer's contribution charge is 11\% of each eligible employee’s ordinary time earnings and is scheduled to progressively grow to its maximum level 12\% from 1 July 2027 onwards. Employer's SG contributions are supplemented by personal contributions of super fund members representing around 20\% of total contributions.
Since the introduction of mandatory superannuation scheme, the proportion of Australian population with superannuation account has reached 78\%, which is one of the highest levels of coverage worldwide. According to the recent data provided by ASFA \cite{ASFA} in September 2023, around 17 million Australians are holders of superannuation account and their total superannuation assets amount to around 140\% of GDP in 2022.

In contrast to other pension systems, Australian super funds distinguish themselves by establishing personal accounts for their members with two main types of super funds: {\it defined benefit super funds} with shared investment risks and limited risk-taking capacity (\cite{D2010}), which are now mostly closed to new members, and nowadays prevailing {\it accumulation super funds}, which grow with contributions and idiosyncratic investment returns. We henceforth focus on the accumulation funds since they enable the provision of highly tailored investment options, accommodating diverse individual beliefs and investment preferences. Accumulation account holders have the flexibility to allocate their investments across various asset classes, including equities, cash investments, property, infrastructure, as well as other assets such as hedge funds and commodities, which also means that their risk-taking capacity is virtually unlimited.

As of September 2023, Australians have collectively accumulated 3.56 trillion AUD, out of which 2.47 trillion AUD held in 1368 APRA-regulated super funds and 0.89 trillion AUD in ATO-regulated self-managed super funds; this makes Australia the fourth largest holder of pension fund assets worldwide. In particular, 1.00 trillion AUD (around 40\% of the total) was invested in MySuper, which is a default option offered by large APRA-regulated super funds.  From the perspective of asset classes, 53.3\% of the total was invested in equities, with a breakdown of 21.9\% in Australian listed equities, 26.4\% in international listed equities, and 5.1\% in unlisted equities. Fixed income and cash investments constituted 28.8\% of total investments, distributed as 20.3\% in fixed income securities (bonds) and 8.5\% in cash (bank bills). Property and infrastructure accounted for 15.6\% of the total, while other asset classes, encompassing hedge funds and commodities, represented 2.2\% of the overall investment portfolio. This nuanced investment allocation strategy underscores the unique nature of Australian superannuation scheme, catering to members' varied expectations and preferences.

Over the past 30 years, the average nominal return of  Australian super funds was 7.2\% and the real return versus CPI was 4.5\%. Despite this overall positive trend, two noteworthy instances of negative returns occurred, first negative 12.9\% return in 2008 during the subprime mortgage crisis of 2007-2010 and then negative 3.3\% return in 2022 amidst the post Covid-19 pandemic high inflation crisis and the ensuing market debacle in the first half of 2022 when the S\&P~500 and ASX~200 (more precisely, the S\&P/ASX~200 index) indices have fallen by 20\% and 12\%, respectively. Not surprisingly, superannuation returns exhibit a high correlation with major market indices due to the predominant inclusion of these indices in investment options offered to members.

\begin{figure} [h!]
    \centering
    \includegraphics[width=16cm, height=9cm]{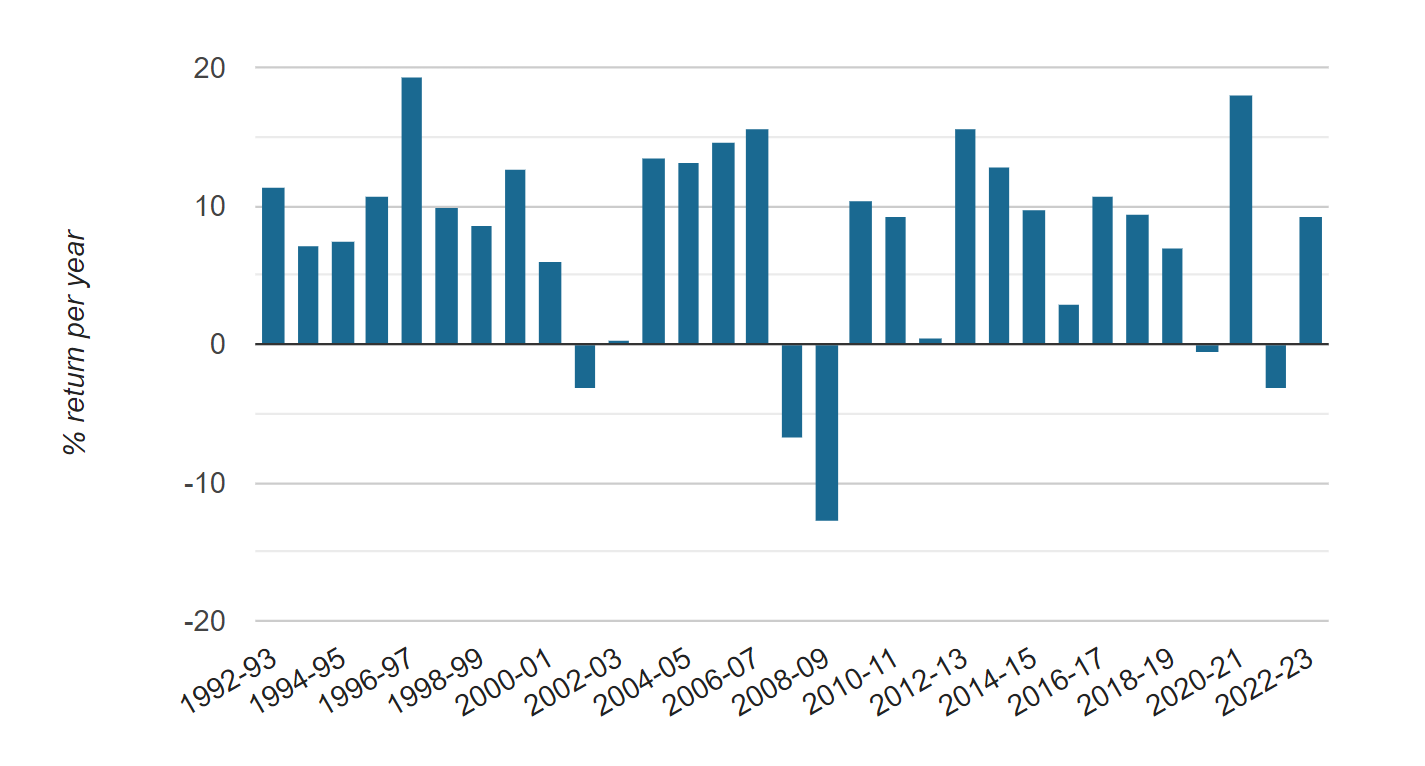}
    \caption{Superannuation nominal returns by year}
    \label{pension_assets}
\end{figure}

One important issue for members of super funds that has not received deserved attention reads: how to manage the risk of a global market downturn? While the account holders often prioritise risk mitigation over pursuing higher returns during the ultimate stage of their accumulation period, they meet a great challenge in finding a desirable investment strategy. This is partially due to the lack of user-friendly financial instruments available to protect their investment returns against an unpredictable risk of a market downturn, especially when in transition to pension phase. This shortage of investment protection products underscores a critical gap in the existing framework for superannuation risk management. We thus believe that it is  important to develop new widely available products serving as an efficient protection against market downturns. By profiting from such standalone
financial instruments, account holders would be able choose more aggressive investment strategies, even during the final stage of accumulation phase
or in pension phase. From a broader financial market perspective, the flow of additional funds from super funds into the market would contribute to sustainable growth of the economy.

Before delving into the Australian superannuation context and introducing a new type of investment insurance, let us first take a brief look at other insurance products offered on pension markets worldwide, especially in the United States where the market for variable annuities has been well established since 1950s. In general terms, a {\it variable annuity} is a financial contract between a policyholder and an insurance company that promises a stream of periodic payments linked to the investment performance of an underlying portfolio.
A typical variable annuity is a unit-linked or managed fund vehicle, which also offers optional guarantee benefits,
commonly referred to as `riders,' as a choice for the customer.
Insurance companies charge periodic fees for managing these annuities and primarily invest in publicly traded securities to create a diversified reference portfolio. They have also responded to the need to protect investors' future income by introducing various additional covenants (riders)
to variable annuities. These riders are designed not only to manage risks, such as market and longevity risks, but also to allow for optimal withdrawal strategies in the face of future uncertainty. Two major types of riders are: the {\it Guaranteed Minimum Death Benefits} (GMDB) and
the {\it Guaranteed Minimum Living Benefits} (GMLB). GMLB riders are further categorised into several classes, including {\it Guaranteed Minimum Income Benefits} (GMIB), {\it Guaranteed Minimum Withdrawal Benefits} (GMWB), and  {\it Guaranteed Lifetime Withdrawal Benefits} (GLWB) but, typically, a single rider combines various features tailored to individual needs. As reported in Drexler et al. \cite{Drexler2017}, there was a significant increase in the sale of variable annuities after the late 1990s when guaranteed minimum benefit riders began to be included in variable annuities and, in 2016, approximately 76\% of policyholders opted to purchase a variable annuity with guarantees, when this option was available.

Regarding the GMDB rider, if the policyholder's death occurs before the contract matures, the beneficiary of the contract will receive the higher of the investment account value or the death benefit. The pricing approach to GMDB riders was first proposed and analysed by Milevsky and Posner \cite{MP2001}. Milevsky and Salisbury \cite{MS2006} examined a variety of static and dynamic methods for pricing GMWB riders. Dai et al. \cite{Dai2008} developed an efficient finite difference algorithm with the penalty approximation to solve the singular stochastic control problem under the optimal withdrawal strategy. Later on, several papers discussed combinations of various riders. In particular, Bauer et al. \cite{Bauer2008} dealt with the valuation of variable annuities with multiple guarantees, with a holder's strategy not only consisting of the decision of surrender option but also possible withdrawals at a set of predetermined dates. Luo and Shevchenko\cite{Luo2015} formulated the pricing problem for GMWDB rider (i.e., GMDB combined with GMWB) as a stochastic control problem by using the conditional expectation of dynamic GMWDB riders for a given death time. Moreover, Yang and Dai \cite{YD2013} and Dai et al. \cite{Dai2015} considered stochastic mortality in GMWB riders and concluded that ignoring mortality risk would lead to an overpricing of the rider. From the existing literature, it is clear that the fair pricing
of variable annuities with riders is a complex optimisation problem, which can be solved by sophisticated providers, but a sound assessment
of their advantages and drawbacks is beyond reach of members of super funds.

As a supplement to compulsory annuities, a new type of insurance product called {\it Registered Index-Linked Annuity} (RILA) has been introduced in the U.S., also known as `index-linked annuities' or `structured variable annuities.'  Moenig \cite{M2021} provided a timely first academic study on RILAs. A RILA is a complex annuity with insurance properties that offer policyholders the flexibility to prioritise their growth opportunities while limiting potential losses. In a generic RILA contract, the policyholder can choose from several index options, protection levels, and other available options to achieve their preferred profile of potential growth. Typically, holders of a RILA can choose to tie their wealth account with the performance of a particular market index, such as the S$\&$P 500 for U.S. large-cap stocks, NASDAQ Composite for all American and foreign common stocks, and the MSCI EAFE for stocks from developed economies outside of North America.
In essence, if the reference index performs poorly, the credited loss is lessened by either a \textit{floor} (i.e., a maximum loss percentage), a \textit{buffer} (i.e., only index losses over a specific threshold are credited), or a \textit{downside participation rate} (e.g., the loss credited to the account corresponds to only 50$\%$ of index losses). As pointed out in \cite{M2021}, the insurer can benefit from the imbalance between the guarantee's downside protection and its upside cap, the ability to invest RILA funds into other products (e.g., corporate bonds) and hence earn a credit risk premium and, in some cases, from using a price-based index not accounting for dividend payments.

It is striking that in today's Australian superannuation system no widely available and functional risk management instruments
similar to RILAs are offered. It should be acknowledged that some kinds of variable annuities with riders are available to members of super funds but they are mostly retirement income products (see, for instance, Shen et al. \cite{SSWZ2023} and the references therein) whereas insurance products for both accumulation and pension phase are missing. Furthermore, as emphasised in \cite{SSWZ2023}, the valuation and assessment of existing products rely on extremely complex numerical techniques and hence their pricing is subject to model risk. This should be contrasted with a robust market-based approach presented here where fair pricing does not require to make use of any stochastic model, provided that suitable options on indices are traded. Most importantly, we propose to completely separate the superannuation account from insurance product, as opposed to the case variable annuities where insurance is embedded in a pension product, for instance, a lifetime pension such a {\it Group Self-Annuity} (GSA), which was proposed by Piggott et al. \cite{PVD2005} and recently reexamined by Shen et al. \cite{SSWZ2023}. Although a proportion of holders of superannuation accounts are likely to invest in variable annuities to support their living expenses in retirement phase, it is also reasonable to expect that not only holders of accumulation accounts but also a group of holders of pensions account may show interest in investment protection products.

In response to the increasing demand for investment insurance by members of super funds, we propose to completely depart from variable annuities with riders and to introduce of a new class of standalone financial products, which we propose to call {\it Equity Protection Swaps} (EPS). Recall that a traditional {\it equity swap} involves the exchange of future cash flows between two parties, allowing each party to diversify their income for a predetermined period while retaining their original assets. An EPS shares some similarities with a common type of equity swaps known as a {\it Total Return Swap} (TRS), which is a financial contract transferring both the credit and market risks of an underlying asset. Due to its design, an EPS can offer effective protection in case the value of a reference portfolio declines, and it can be easily tailored to meet the specific needs of super fund members.

It is important to note that the proposed approach is completely independent of the issue of optimal investment strategies since we take for granted that the majority of pension account holders lack the necessary financial or mathematical expertise to construct an optimal portfolio and hence they are choosing the default option, such as MySuper. Moreover, the current superannuation system in Australia is characterised by a relatively limited selection of investment options offered by super funds, restricting investors from completely freely choosing their desired investments. Given the limited understanding of investment principles among most super funds members, pursuing an optimal investing approach to superannuation products would be clearly unrealistic and of no practical importance. At the same time, there remains a significant, and intuitively clear, desire among investors to mitigate their exposure to investment risk. This underscores the relevance of Equity Protection Swaps, which are tailored to provide an effective and relatively simple way of risk reduction (investment insurance) to holders of superannuation account.

When compared to the existing insurance products for variable annuities, an EPS offer several advantages. Firstly, it has a more transparent structure than a typical variable annuity rider, making it easier for super funds members to assess and appreciate. Second, it does not require from its holder to make any decisions during the contract's lifetime. In particular, there is no need to consider an embedded optimal stopping problem, which is a feature commonly encountered in most riders. Thirdly, the fair pricing of an EPS relies on static perfect hedging using liquidly traded options. Hence it is expected that an EPS would be a straightforward and cost-effective tool for insurance providers and holders of superannuation accounts. This should be contrasted with opaque and extremely complex utility-based pricing methods for riders such as GMWB, GMIB or GSA, typically with no reference to effective hedging for providers, although a martingale measure for equities model is sometimes formally utilised. Most importantly, the fair pricing of a basic EPS can be considered to be model-free since market data for traded European options can be used to compute its fair premium.

Finally, an EPS is not inherently tied to any annuity and hence it can be issued by a third party, which is here referred to as a provider of investment insurance.  In fact, the key distinction between a RILA and an EPS is that the latter is a standalone financial derivative, which can be added to a superannuation account at any phase (not only after retirement), whereas RILA is a genuine variable annuity, which can be purchased by retirees using funds from their pension accounts. In the context of the Australian superannuation system, it should be stressed that members are not allowed to withdraw funds from their superannuation accounts during the accumulation phase to purchase an investment guarantee such as, e.g., a RILA. Instead, members are compelled to select a combination of investment options, which are offered by a super fund, and thus they are entirely responsible for managing market risks. An EPS is tailored to serve as a convenient risk mitigation tool for super funds members by offering a highly desirable and transparently designed protection against potential losses at a minimal cost, which is only incurred when their investment portfolio yields substantial gains.

The paper is organised as follows. We first present in Section \ref{sec2} the definition and main features of a generic \textit{EPS}. We propose a classification of the fee and protection legs of an EPS with special attention to \textit{buffer} and \textit{floor} legs.
Section \ref{sec3} is devoted to an analysis of hedging strategies for equity protection swaps. We show in Proposition \ref{prop5.2.1} that
an EPS admits a static hedge composed of European call and put options and thus its fair price can be computed in
any arbitrage-free model of the financial market. Even more importantly, the fair price of an EPS can also be obtained directly from market data for traded European call and put options if such options are available for the reference portfolio (usually a market index). Then the static hedge for an EPS can be easily constructed by investing in traded options and hence the hedge is independent of a model. If, however,
European options needed for hedging are not available, then their synthetic counterparts can be produced through a dynamic trading strategy using a
preferred dynamical model for the reference portfolio (e.g., the classical Black-Scholes \cite{BS1973} model, the Merton \cite{Merton1976} jump-diffusion model, the Heston \cite{Heston1993} stochastic volatility model, or the Bates \cite{Bates1996} model).

In Section \ref{sec4}, we present forward performance tests of buffer and floor EPSs within the framework of the classical Black-Scholes model using the pricing formula for a generic EPS obtained in Proposition \ref{BSEc}.
In the first part of model-based numerical studies, we use buffer and floor EPSs introduced before as examples of a generic EPS. By specifying interest rates, volatility, and other relevant market parameters, we compute the fair initial premium by using static hedging strategies for an EPS.
We can see the impact of different characteristics of an EPS, such as participation rate, leg setting, and maturity, on fair premiums. It is widely acknowledged that a typical investor would not be willing to pay an upfront premium, especially when it is substantial. Therefore, we propose to focus on EPS products with a null fair premium at inception. An EPS with a null fair premium still has the  participation rate for the protection leg desired by the holder, whereas the participation rate for the fee leg is chosen to ensure that the provider
can implement a hedging strategy based on long and short positions in European options at no additional cost.

Unlike Section \ref{sec4} where all results are model-based, Section \ref{sec5} is devoted to backtesting of real-world benefits for holders of EPSs using historical data for S\&P~500 and S\&P/ASX~200 indices and hence results and conclusions are based on model-free market data. We ponder over two situations during a period under study: a market crash generating substantial losses and a general market characterised by moderate gains. An analysis of the first situation shows that an EPS can provide highly effective protection against large losses, whereas the second situation illustrates the tradeoff between loss protection and profit sharing. The initial premia for EPS products used in our analysis can be obtained directly from the market data for European call and put options, as opposed to calculations based on a hypothetical model. The conclusion from
backtesting is that an EPS can serve as an efficient tool to mitigate financial risks associated with a superannuation account.

Finally, in Section \ref{sec6} we briefly summarise our findings and we give an outline possible directions for future research.
In particular, we emphasise the importance of a cross-currency variant of an equity protection swap for holders of superannuation
accounts in Australia.

\section{Properties and Classification of Equity Protection Swaps} \label{sec2}

In this section, we will build a general pricing model for adjusted return of the underlying reference portfolio for an {\it equity protection swap} (EPS). By an {\it EPS}, we mean a financial contract in which a single terminal payoff occurs at maturity date $T$ and is a predetermined function of the realised simple rate of return on a {\it reference portfolio} during the lifetime of an EPS. For instance, a reference portfolio may represent the value of a particular market index, a combination of domestic and foreign indices or a bespoke portfolio of domestic and foreign equities. 

\subsection{Specification of a Generic Equity Protection Swap} \label{sec2.1}

Let $S_t$ stand for the value at time $t \in [0,T]$ of a given reference portfolio of assets, which is denoted as $S$.
The simple rate of return on $S$ over the time period $[0,T]$ is denoted as $R_S(0,T)$ but, for brevity, we will also write $R_T$ instead of $R_S(0,T)$ so that $R_T=R_S(0,T):=(S_T-S_0)/S_0$  and the inequality $R_T \geq -1$ holds under our standing assumptions that $S_0$ is a strictly positive number and $S_T$ is a nonnegative random variable. Formally, the stochastic process $S$ is defined on a probability space $(\Omega,\cF,\ff,\bbP)$ endowed with the filtration $\ff$ where $\bbP$ is the statistical probability and thus the rate of return $R_T$ is an $\cF_T$-measurable random variable.

Let us consider a generic equity protection swap so that an explicit specification of a reference portfolio $S$ is not important at this stage.
Suppose that a holder of a superannuation account decides at time 0 to invest some of their funds in $S$ and let $N_p$ denote the {\it nominal principal} of their holdings that are invested (either directly or indirectly) in $S$ at time 0. Then the dynamics of their funds that are invested in $S$ will follow the fluctuations of the value of a portfolio $S$ and thus the current market value of their funds invested in $S$ will grow (if $R_T>0$) or decline (if $R_T<0$) to the terminal value $N_p R_T$ at time $T$.

\newpage

As was already mentioned, the goal of an EPS is to reduce the holder's risk exposure by providing either a full or partial protection against a substantial decline of holder's wealth invested in $S$ by entering into a customised contract offered by a financial company. The two counterparties of an EPS are henceforth referred to as the {\it protection provider} who has the short position in an EPS and the {\it holder} of an EPS (the protection buyer or, simply, the investor) who has the long position in the contract. It should be stressed that aims of the two counterparties
are completely different: the holder is reducing their risk exposure by partially `selling' it to the protection provider, whereas the financial company accept the holder's risk and their goal is to completely hedge their risk exposure using traded financial derivatives (European options on $S$). By analysing available hedging strategies and market prices of options, the company is also able to provide a fair valuation of an EPS in the form of the fee participation rate in holder's gains.

From the viewpoint of the protection provider, the cash flow at time $T$ is given by $N_p \pp (R_T)$ where $N_p$ is the nominal principal of an EPS and a function $\pp : [-1,\infty) \to \bbR$ is used to encode the structure of an EPS.
By the {\it adjusted return} we mean a specific transformation $\pp (R_T)$ of the realised return $R_T$, as given by the covenants of each particular EPS. Before offering a product, the provider of an EPS needs to examine the pricing and hedging problem although, formally, a solution to the pricing and hedging problem is symmetric if the bid-ask spread for traded options is neglected and thus, at least in principle, an investor can also derive that solution since it does not require to use sophisticated mathematical tools.

From the holder's perspective, it is important to notice that the net payoff to the buyer of an EPS is given by $N_p(R_T-\pp(R_T))$ or, equivalently, it equals $R_T-\pp(R_T)$ per one unit of the nominal principal. The benefit from entering an EPS for investors lies in reduction of losses and thus the transformation $\pp$ should be specified in accordance with their individual investment strategies and risk preferences. Therefore, several variants of an EPS are likely to be offered to investors and some of them are studied in the present work.
However, the main features shared by all EPSs are: a desirable reduction of the holder's loss, the null initial cost,
and a fair protection fee in the form of partial sharing of the holder's gain from their investment in $S$. It is clear that an EPS is a standalone
financial derivative written on $S$ and the values of $N_p$ and $S_0$ are unrelated.

We are in a position to formalise the definition of a generic EPS. Unless explicitly stated otherwise, we henceforth assume, without loss of generality, that the notional principal of an EPS equals 1, that is, $N_p=1$. In fact, due to specification of a generic EPS, one could also
assume that $S_0=1$ since only the return $R_T$ is relevant but for clarity of presentation we are not going to make that assumption.
Since it is natural to postulate that an EPS can be hedge using plain-vanilla call and put options, we postulate that the payoff profile $\pp :[-1,\infty)\to \bbR$ of a generic EPS is a piecewise linear, non-decreasing, continuous function such that $\pp (0)=0$ where the last condition is natural.

\begin{definition} \label{EPS} {\rm
In an \textit{equity protection swaps} starting at time $0$ with the maturity date $T$ and the reference portfolio $S$,
the cash flow for the EPS provider is specified by the \textit{adjusted return} $\pp (R_T)$ where the derivative of the \textit{adjusted return function} $\pp:[-1,\infty)\to\mathbb{R}$ equals
\begin{equation} \label{eq0.2.1}
\pp'(R)=\sum_{k=-(n+1)}^{m}\alpha_{k+1}\I_{(\beta_k,\beta_{k+1})}(R)
\end{equation}
where the sequence of endpoints of return ranges, $(\beta_k)_{k=-n}^{m}$, is a strictly increasing finite sequence of real numbers, $\beta_{-(n+1)}=\lbn$ and $\beta_{m+1}=\infty$. The sequence $(\alpha_k)_{k=-n}^{m+1}$ of marginal participation rates is an arbitrary sequence of numbers from $[0,1]$. Furthermore, we assume that $\beta_0=0$ and thus $\lbn=\beta_{-(n+1)}< \beta_{-n}<\cdots <\beta_{-2}< \beta_{-1}<\beta_{0}=0<\beta_{1}<\cdots <\beta_{m}<\beta_{m+1}=\infty$.}
\end{definition}

Consistently with its financial interpretation, we expect that the number $\alpha_k$ will belong to the interval $[0,1]$ since it represents the provider's marginal participation rate in gains or losses of the reference portfolio when the realised return falls within the range of $\beta_k$ to $\beta_{k+1}$ (though in some exceptional cases the fee participation rate can be greater than 1).
In particular, $\alpha_k = 1$ corresponds to the {\it full participation} of the provider in gains or losses, whereas $\alpha_k \in (0,1)$ represents a {\it partial participation} in gains (upside participation rate) or losses (downside participation rate).
Of course, $\alpha_k=0$ means that the provider does not participate in incremental returns belonging to the interval $[\beta_k,\beta_{k+1}]$. Since $\alpha_k$ represents the provider's participation rate, it is clear that the corresponding buyer's participation rate equals $1-\alpha_k$ for every $k$.

\newpage

When analyzing particular instances of fee and protection legs of an EPS, it will be convenient to modify the notation by independently labeling the negative and positive values of $\beta$. Specifically, we write $\lbn = l_{n+1}< l_{n}<\cdots <l_{2}< l_{1}<l_0=0$ for non-positive values of $\beta$ and $0=g_{0}< g_{1}< \cdots <g_{m}< g_{m+1} = \infty$ for non-negative values of $\beta$. For each $k$, the participation rate over the loss interval $[l_{n+1},l_n)$ is denoted as $p_{n+1}$ whereas the participation rate over the gain interval $(g_{k},g_{k+1}]$ is denoted as $f_{k+1}$.

In a simple example where $n=m=1$, we have $\lbn =l_2<l_1<l_0=0=g_{0}<g_1<g_2=\infty$ and we need to specify four values of marginal
participation rates: $p_1$ and $p_2$ for the protection leg and $f_1$ and $f_2$ for the fee leg. It is thus clear that we deal with six parameters, which fully specify the structure of an EPS when $n=m=1$. The following lemma is an easy consequence of \eqref{eq0.2.1} and the postulated equalities $\pp (0)=\pp (g_0)=\pp (l_0)=0$.

\begin{lemma} \label{lem0.2.1}
We have that $\pp (R)\leq 0$ for $R \in [\lbn ,0]$ and $\pp (R)\geq 0$ for $R \in [0,\infty)$. The {\it protection leg} and the {\it fee leg}
of an EPS are given by
\begin{equation*}
\pp^p(R):=\pp(R)\I_{\{R<0\}}, \quad \pp^f(R):=\pp(R)\I_{\{R>0\}}
\end{equation*}
so that $\pp (R_T)=\pp^p(R_T)+\pp^f(R_T)$ where the non-positive payoff $\pp^p(R_T)$ and the non-negative payoff $\pp^f(R_T)$ represent
the provider's loss (protection payout) and gain (fee income) from an EPS, respectively.
\end{lemma}

\begin{proof}
Let us define $\Delta l_i=l_{i+1}-l_i<0$ and $\Delta g_j=g_{j+1}-g_j>0$ and recall that $g_0=l_0=0$. Then the protection and fee legs satisfy,
for $i=0,1,\dots,n$  and $j=0,1,\dots,m$,
\begin{align*}
\pp^p(R)&=\Big(\sum_{k=0}^{i-1}p_{k+1}\Delta l_k+p_{i+1}(R-l_i)\Big)\I_{\{R \in (l_{i+1},l_i)\}}, \\
\pp^f(R)&=\Big(\sum_{k=0}^{j-1}f_{k+1}\Delta g_k+f_{j+1}(R-g_j)\Big)\I_{\{R \in (g_j,g_{j+1})\}},
\end{align*}
and it is clear that the adjusted return function of an EPS satisfies $\pp(R)=\pp^p(R)+\pp^f(R)$.
\end{proof}

Lemma \ref{lem0.2.1} leads to the following definition of {\it fairness} of a generic EPS, which explains why the term `swap' is suitable for a generic EPS. Although, by convention, the inception date of an EPS is chosen to be 0, it can be substituted with any date $t<T$.

\begin{definition} \label{EPSfair} {\rm
We say that an EPS initiated at time 0 with a payoff structure $\pp$ is {\it fair} if  the payoff
$\pp (R_T)$ with settlement date $T$ has null arbitrage-free price at its inception time 0.}
\end{definition}

It will be shown in Section \ref{sec3} that the issue of finding a fair structure for an EPS is closely related to the pricing problem
for European options with the underlying asset $S$ and maturity date $T$.

Let us make some comments on Definition \ref{EPS} and Lemma \ref{lem0.2.1}: \begin{itemize}
\item Since the main goal of an EPS is to protect against substantial losses, the loss leg should be chosen by the holder from several available options. As soon as the holder's choice is completed, the issue of `fair pricing' of an EPS hinges on finding parameters for the fee leg that make the contract worthless at its inception.
\item In a typical structure, the fee leg $\pp^f$ (resp., the protection leg $\pp^p$) can be formally identified with the payoff of a particular portfolio of European call options (resp., European put options) written on the underlying static portfolio (e.g., a market index).
\item Suppose now that the underlying portfolio is dynamic, that is, the portfolio is actively managed by a holder of a superannuation account.
    If a holder of an EPS on a given portfolio wishes to switch at some moment before $T$ from one investment option to another, then an EPS can be terminated and settled at its market price and a new EPS can be initiated
    at no cost if this is desired by a holder. This means that a flexible roll-over strategy in EPSs is available since at any time the market value of an EPS is readily available due to the fact that the hedging portfolio is composed of traded options.
\item One could also consider an EPS with a discontinuous adjusted reward function $\pp$ but it is not clear if this would be beneficial for practical purposes and, obviously, the hedging portfolio for an EPS would require options with discontinuous payoff profiles.
\end{itemize}

Despite some similarities, there are several crucial differences between an EPS and a RILA. Firstly, the main difference is that a RILA is special case of a variable annuity while an EPS is a standalone financial derivative (a swap), which is designed as
a complement to a superannuation account. The provider of RILA receives the nominal principal from the investor at the start of contract, whereas there should be no cash flow at the inception of an EPS, of course, provided that it is set up and priced correctly. Although representations of the provider's gains and losses in terms of European options are analogous to those obtained for RILAs in Section 3.2 of Moenig \cite{M2021}, their analysis is different since the provider of an EPS does not need to manage the nominal principal and thus has lower expenses and risks.

Secondly, an EPS has typically a more flexible structure than a RILA, for which one can only consider an existing variable annuity, which is provided by an insurance company and thus always has a fixed structure with either floor, buffer or cap and some predetermined participation rates. In contrast, an EPS is a bespoke financial derivative so its structure can be brokered by the two parties of the contract. Thirdly, the underlying portfolio of the RILA is also fixed (a particular market index), whereas the reference portfolio of an EPS can be chosen arbitrarily, for instance: a market index, an bespoke equity portfolio, or an investment portfolio involving assets from several economies and, as was explained above, if the reference portfolio is changed before the nominal maturity $T$ then an EPS can be terminated by its holder and settled in reference to its current market value.

\subsection{Basic EPS with Proportional Fee and Protection Legs} \label{sec2.2}

Let us assume $n=m=1$ and thus $\lbn < l_{1}<l_0=0=g_{0}< g_{1}< \infty$. The most basic fee leg relies upon setting $f:=f_1=f_2 \in (0,1]$, which means that the provider is awarded proportional participation in all positive returns on a reference portfolio $S$. Similarly, the simplest protection leg is obtained by taking $p:=p_1=p_2 \in (0,1]$, which means the provider has a proportional participation in all negative returns. Then the cash flow $\pp(R_T)$ equals
\begin{align*}
\pp(R_T)=p R_T\I_{\{R_T \in [\lbn,0)\}}+f R_T\I_{\{R_T \in (0,\infty)\}}=-p(-R_T)^+ + f(R_T)^+
\end{align*}
and hence it can be represented by a short position in $p/S_0$ units of a European put combined with a long position in $f/S_0$ units of a
European call written on $S$ with strike $S_0$ and maturity $T$ since
\begin{align*}
-p(-R_T)^+ = -\frac{p}{S_0}\,(S_0-S_T)^+, \quad f(R_T)^+ = \frac{f}{S_0}\, (S_T-S_0)^+.
\end{align*}

We will now show that the basic EPS can be fairly priced in relation to market prices of European call and put options.
Let $\text{Call}_t(K,T)$ and $\text{Put}_t(K,T)$ denote the price at time $t$ of a European call and put on
the reference portfolio with strike $K$ and maturity $T$, respectively. Then
\begin{align*}
\pp (R_T)=\frac{f}{S_0}\,\text{Call}_T(S_0,T)-\frac{p}{S_0}\,\text{Put}_T(S_0,T)
\end{align*}
and thus the fair value of the EPS at time 0 for its provider equals
\begin{align*}
\frac{f}{S_0}\,\text{Call}_0(S_0,T)-\frac{p}{S_0}\,\text{Put}_0(S_0,T).
\end{align*}
An application of the put-call parity, which reads $\text{Call}_0(K,T)-\text{Put}_0(K,T)=S_0 - Ke^{-rT}$, gives the equality $\text{Call}_0(S_0,T)-\text{Put}_0(S_0,T)=S_0(1-e^{-rT})$
and thus yields the following representation for the fair value for the provider of the basic EPS
\begin{align} \label{eq2.30s}
&  \frac{f}{S_0}\,\text{Call}_0(S_0,T)-\frac{p}{S_0}\,\text{Put}_0(S_0,T)=\frac{f}{S_0}\,\text{Call}_0(S_0,T)
  -\frac{p}{S_0}\Big(\text{Call}_0(S_0,T)-S_0\big(1-e^{-rT}\big)\Big)\nonumber  \\ &= \frac{f-p}{S_0}\,\text{Call}_0(S_0,T)+p\big(1-e^{-rT}\big)=0.
\end{align}

Notice that the last equality in \eqref{eq2.30s} is not trivially satisfied for any choice of $f$ and $p$ but is postulated here for consistency with the swap requirement that there should be no cash flow in the basic EPS at time 0 and thus its fair value at time 0 should vanish.
It is clear that if the risk-free rate $r$ and the market price $\text{Call}_0(S_0,T)$ are known, then it is possible to find a relationship between the parameters $f$ and $p$ that makes the basic EPS worthless at time 0. Finally, the theoretical equality $\text{Call}_0(S_0,T)=S_0\,\text{Call}_0(1,T)$ holds and thus the relationship between the fair values of $f$ and $p$ reduces to $(f-p)\,\text{Call}_0(1,T)+p(1-e^{-rT})=0$. It is clear that the equality $f=p$ holds if $r=0$.

In contrast, if $r>0$ then the basic EPS with $f=p$ is not fair since an EPS allows the provider to participate in returns on a portfolio $S$ without actually investing in $S$ his wealth, presumably borrowed at a positive interest rate $r$. Of course, it becomes a fair EPS for some values $f>p$, which can be computed from \eqref{eq2.30s}.  An analogous analysis can be conducted for the basic EPS in the less commonly encountered case where $r<0$.

\subsection{Classification of Fee and Protection Legs} \label{sec2.3}

In this subsection, we will usually assume that $n=m=1$ (except for the \textit{buffer-cap} and \textit{buffer-floor} legs) and thus $\lbn < l_{1}<l_0=0=g_{0}< g_{1}< \infty$. Then each particular structure of an EPS depends on a specification of four participation rates, namely, $p_1$ and $p_2$ for the protection leg and $f_1$ and $f_2$ for the fee leg. For brevity, the parameters $p_1$ and $p_2$ (resp. $f_1$ and $f_2$) are henceforth referred as {\it protection rates} (resp. {\it fee rates}).

Let us first describe the most typical and practically appealing instances of a protection leg. Although all four cases can be considered,
we contend that the buffer (and perhaps also buffer-cap) case would be the most suitable, especially when
combined with an analogous structure for the protection leg. Generally speaking, the goal of an EPS is to give protection
against substantial losses so it is natural to expect that there will be no cash flow at settlement if the realised return falls within the range $(l_1,g_1)$. Only when realised losses are below the trigger $l_1<0$ a buffer (or a buffer-cap) EPS provides a full or partial protection
to its holder who, in exchange, agrees to share with the provider portfolio's gains if they fall above a predetermined level $g_1>0$.

\begin{definition} \label{Fee} {\rm
The \textit{proportional fee leg} is given by $f_1=f_2 \in (0,1]$ so that
\begin{align*} 
\pp^{f,1}(R)=f_1 R\I_{\{R \in (0,\infty)\}}.
\end{align*}
The \textit{buffer fee leg}  is given by $f_1=0$ and $f_2 \in (0,1]$ so that
\begin{align*} 
\pp^{f,2}(R)=f_2(R-g_1)\I_{\{R \in (g_1,\infty)\}}.
\end{align*}
The \textit{cap fee leg}  is given by $f_1\in (0,1]$ and $f_2=0$ so that
\begin{align*}  
\pp^{f,3}(R)=f_1 R \I_{\{R \in (0,g_1]\}}+f_1 g_1\I_{\{R \in (g_1,\infty)\}}.
\end{align*}
The \textit{buffer-cap fee leg}  is given by $f_1=0, f_2\in (0,1]$ and $f_3=0$ so that}
\begin{align*} 
\pp^{f,4}(R)=f_2(R-g_1)\I_{\{R \in (g_1,g_2]\}}+f_2(g_2-g_1)\I_{\{R \in (g_2,\infty)\}}.
\end{align*}
\end{definition}

The most typical specifications of the protection leg are analogous to those of the fee leg. A selection of a particular protection leg depends on the buyer's preferences and thus it would be natural to expect that a broad spectrum of products should be offered by EPS providers.
Assuming that a perfect hedge of an EPS is feasible, the provider would be indifferent with respect to the buyer's choice of the structure of the protection leg. However, in reality only a partial hedging can be attained for more complex cross-currency products and thus some forms of the protection leg are likely to be preferred by providers. For instance, the presence of a floor clause is expected to be appreciated by providers since it provides a lower bound on their downside risk exposure even when the market experiences a catastrophic downturn.

\begin{definition} \label{Prot}
{\rm The \textit{proportional protection leg} is given by $p_1=p_2 \in (0,1]$ so that
\begin{align*} 
\pp^{p,1}(R)=p_1 R\I_{\{R \in [\lbn,0)\}}.
\end{align*}
The \textit{buffer protection leg} is given by $p_1=0$ and $p_2 \in (0,1]$ so that
\begin{align*}  
\pp^{p,2}(R)= p_2(R-l_1) \I_{\{R \in [\lbn,l_1)\}}.
\end{align*}
The \textit{floor protection leg} is given by $p_1\in (0,1]$ and $p_2=0$ so that
\begin{align*} 
\pp^{p,3}(R)= p_1R\I_{\{R \in [l_1,0)\}}+p_1 l_1 \I_{\{R\in [\lbn,l_1)\}}.
\end{align*}
The \textit{buffer-floor protection leg} is given by $p_1=0, p_2\in (0,1]$ and $p_3=0$ so that}
\begin{align*} 
\pp^{p,4}(R)=p_2(R-l_1)\I_{\{R \in [l_2,l_1)\}}+p_2(l_2-l_1)\I_{\{R \in [\lbn,l_2)\}}.
\end{align*}
\end{definition}

\subsubsection{Buffer EPS} \label{sec2.3.1}

Let us introduce two most practically relevant forms of an EPS, which are called the \textit{buffer EPS} and the \textit{floor EPS}. Notice that the proposed terminology for a generic EPS is referring directly to the protection leg, rather than the fee leg for which the choice of a buffer
structure seems to be the most appealing for holders. It should be observed that the net payoff to the holder of a portfolio $S$ and an EPS is given by $R_T-\pp(R_T)$ per one unit of the nominal principal and thus it can be easily computed for any choice of $\pp$.

We first define the \textit{buffer EPS} where, incidentally, the concept of a buffer leg is applied to both the protection and fee legs and thus it could also be called a \textit{double-buffer EPS} or a \textit{buffer/buffer EPS}.

\begin{definition} \label{BEPS}
{\rm A \textit{buffer EPS} is obtained by taking $p_1=f_1=0$ and any values for $p_2\in (0,1]$ and $f_2\in (0,1]$. Hence
$\pp_B (R)=0$ for all $R \in [l_1,g_1]$ where $l_1<0$ and $g_1>0$ and the payoff at maturity $T$ for the provider of a buffer EPS equals}
\begin{align*}
\pp_B (R_T)& := \pp^{p,2}(R_T)+\pp^{f,2}(R_T) = p_2(R_T-l_1) \I_{\{R_T \in [\lbn,l_1)\}}+f_2(R_T-g_1)\I_{\{R_T\in (g_1,\infty)\}}\\
& = -p_2(l_1-R_T)^+ +f_2(R_T-g_1)^+.
\end{align*}
\end{definition}

It is clear that the provider of the buffer EPS only participates in negative returns below the buffer threshold $l_1$ and the positive return above the cap threshold $g_1$ with respective participation rates $p_2$ and $f_2$.  The extremal case where $p_2=1$ and $f_2=1$ corresponds to the provider's full participation in portfolio's losses below the buffer threshold $l_1<0$ combined with his full participation in gains above the cap $g_1>0$.

The adjusted return function $\pp_B$ for the buffer EPS that the provider receives at maturity $T$ is presented in Figure \ref{buffer EPS} where the horizontal axis represents the return of the underlying reference portfolio at maturity and the vertical axis represents the provider's gain or loss. Notice that the protection rate $p_2$ and the fee rate $f_2$ give the slopes of the adjusted return on a reference portfolio.

We will argue that the fair price of the buffer EPS can be easily expressed in terms of prices of European put and call options on a reference portfolio (e.g., a particular investment choice or a market index) with expiry $T$ and suitable nominal values and strikes. If $l_1$ and $g_1$ are given, then for any value of the protection rate $p_2$ chosen by the buyer, one can compute a unique level of the fee rate $f_2$ for which the fair value of the buffer EPS at its inception is null.  


\begin{figure}\centering
\setlength{\unitlength}{1.8mm}
\begin{picture}(24,24)(-10,-10)
\put(0,0){\vector(1,0){27}}
\put(0,0){\vector(0,1){11}}
\put(0,0){\line(-1,0){27}}
\put(0,0){\line(0,-1){10}}
\put(12,-0.5){{\line(0,1){1}}}
\put(-10,-0.5){{\line(0,1){1}}}
\put(-23,-0.5){{\line(0,1){1}}}
\thicklines
\put(0,0){{\color{blue}\line(1,0){12}}}
\put(12,0){{\color{blue}\line(11,9){14}}}
\put(0,0){{\color{blue}\line(-1,0){10}}}
\put(-10,0){{\color{blue}\line(-10,-9){13}}}
\put(-8,12){\scriptsize adjusted realised return $\pp_B(R_T)$}
\put(20,-2){\scriptsize realised return $R_T$}
\put(-24,-2){\scriptsize $-1$}
\put(-10,-2){\scriptsize $l_1$}
\put(11,-2){\scriptsize $g_1$}
\put(-7,1){\scriptsize $p_1=0$}
\put(-17.5,-5){\scriptsize $p_2$}
\put(4.5,1){\scriptsize $f_1=0$}
\put(18,4,5){\scriptsize $f_2$}
\end{picture}
\caption{Provider's adjusted return for a buffer EPS} \label{buffer EPS}
\end{figure}
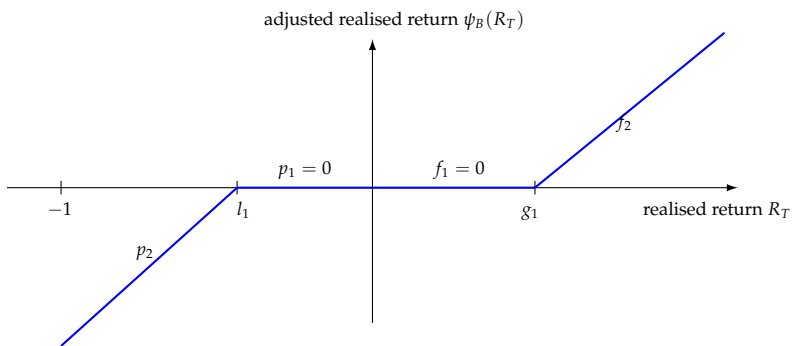

\subsubsection{Floor EPS} \label{sec2.3.2}

The so-called {\it floor EPS}, which can also be called {\it floor/buffer EPS}, is obtained by setting $p_2=0,p_1 \in (0,1],f_1=0$ and $f_2 \in (0,1]$, which means that the provider fully covers the buyer's losses above the level $l_1$ if $p_1=1$ but the buyer's losses below $l_1$ are only partially covered since $p_1 l_1$ acts as a floor for the adjusted return. By convention, the buffer specification is applied to the fee leg and thus the fee leg is exactly the same as in the buffer EPS of Definition \ref{BEPS}. This choice for the fee leg is motivated by our practical arguments that the buffer fee leg is likely to be appreciated by a buyer of an EPS since there will be no fee payment at all, unless the portfolio's gains during the lifetime of an EPS are substantial. The adjusted return function $\pp_F(R_T)$ for the floor EPS is shown in Figure \ref{floor EPS}.

\begin{definition} \label{FEPS}
{\rm A \textit{floor EPS} is specified by parameters $p_2=f_1=0$ and arbitrary values of $p_1 \in (0,1]$ and $f_2 \in (0,1]$.
Hence $\pp_B (R)=0$ for all $R \in [0,g_1]$ where $g_1>0$ and the payoff for the provider equals}
\begin{align*}
\pp_F (R_T)&:= \pp^{p,3}(R_T)+\pp^{f,2}(R_T)=p_1 l_1\I_{\{R_T \in [\lbn,l_1]\}}-p_1(-R_T)^+\I_{\{R_T \in (l_1,0]\}}+f_2(R_T-g_1)^+
\\ &=-p_1 (-R_T)^++p_1(l_1-R_T)^++f_2(R_T-g_1)^+.
\end{align*}
\end{definition}


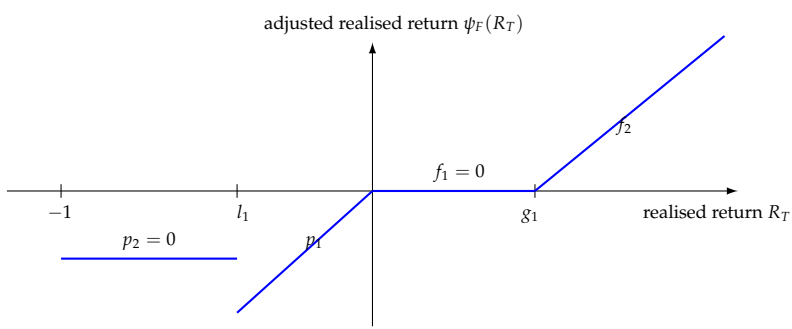
\begin{figure}\centering
\setlength{\unitlength}{1.8mm}
\begin{picture}(24,22)(-10,-10)
\put(0,0){\vector(1,0){27}}
\put(0,0){\vector(0,1){11}}
\put(0,0){\line(-1,0){27}}
\put(0,0){\line(0,-1){10}}
\put(12,-0.5){{\line(0,1){1}}}
\put(-10,-0.5){{\line(0,1){1}}}
\put(-23,-0.5){{\line(0,1){1}}}
\thicklines
\put(0,0){{\color{blue}\line(1,0){12}}}
\put(12,0){{\color{blue}\line(11,9){14}}}
\put(0,0){{\color{blue}\line(-10,-9){10}}}
\put(-10,-5){{\color{blue}\line(-1,0){13}}}
\put(-8,12){\scriptsize adjusted realised return $\pp_F(R_T)$}
\put(20,-2){\scriptsize realised return $R_T$}
\put(-24,-2){\scriptsize $-1$}
\put(-10,-2){\scriptsize $l_1$}
\put(11,-2){\scriptsize $g_1$}
\put(-5,-4){\scriptsize $p_1$}
\put(-18.5,-4){\scriptsize $p_2=0$}
\put(4.5,1){\scriptsize $f_1=0$}
\put(18,4.5){\scriptsize $f_2$}
\end{picture}
\caption{Provider's adjusted return for a floor EPS}\label{floor EPS}
\end{figure}


\subsubsection{Examples of a Buffer and Floor EPS} \label{sec2.3.3}

In Section \ref{sec4}, we present selected numerical results for some EPSs with different payoff structures.
We refer, in particular, to Table \ref{table:EPS} in Section \ref{sec4} for the values of fair premia for EPSs considered in Examples \ref{buffer_ex} and \ref{floor_ex}. It should be noted that an EPS should be designed in such a way that its initial fair premium is null and this is clearly not satisfied in basic examples given below (see, however, Section \ref{sec4.2}).

\begin{example} \label{buffer_ex} {\rm
Consider a buffer EPS with the nominal principal one million AUD (\$1 mm), the buffer threshold $l_1=-5\%$ with the protection rate $p_2=80\%$, the fee threshold $g_1=10\%$ with the fee rate $f_2=80\%$, and maturity one year. Its fair premium within the Black-Scholes setup of Section \ref{sec4} amounts to \$1,168 (see the buffer EPS No. 9 in Table \ref{table:EPS}) and, in three sample scenarios, the buffer EPS is settled at time $T$ in the following way:
\begin{itemize}
\item If the realised rate of return $R_T$ is above $10\%$, say $12\%$, then the buyer pays $(12\% - 10\%) \times 80\% \times \$1\ \mbox{\rm mm} = \$16$ k to the provider.
\item If $R_T$ is between $-5\%$ and $10\%$, then there is no cash exchange between the buyer and provider.
\item If $R_T$ is below $-5\%$, say $-8\%$, then the buyer receives $(-8\% + 5\%) \times 80\% \times \$1\ \mbox{\rm mm}  = \$24$ k from the provider.
\end{itemize} }
\end{example}

\begin{example} \label{floor_ex} {\rm
Consider a floor EPS with the nominal principal \$1 mm, the floor threshold $l_1=-15\%$ with the protection rate $p_1=80\%$, the fee threshold $g_1=10\%$ with the fee rate $f_2=80\%$, and maturity one year. Its fair premium within the Black-Scholes setup equals \$4,276 (see the floor EPS No. 13 in Table \ref{table:EPS}) and the floor EPS is settled at time $T$ as follows:
\begin{itemize}
\item If the realised rate of return $R_T$ is above $10\%$, say $12\%$, then the buyer pays $(12\% - 10\%) \times 80\% \times\$1\ \mbox{\rm mm} = \$16$ k to the provider.
\item If $R_T$ is between $0\%$ and $10\%$, then there is no cash exchange between the buyer and provider.
\item If $R_T$ is between $-15\%$ and $0\%$, say $-8\%$, then the buyer gets $8\% \times 80\% \times \$1\ \mbox{\rm mm}  = \$64$ k from the provider.
\item If $R_T$ is below $-15\%$, say $-20\%$, then the buyer receives $15\% \times 80\% \times \$1\ \mbox{\rm mm}  = \$120$ k from the provider. This is the maximum amount that the buyer may receive, no matter how large is her actual loss above $15\%$ cap.
\end{itemize}}
\end{example}

\subsection{Classification of Equity Protection Swaps}    \label{sec2.4}

Let us make an attempt to classify Equity Protection Swaps according to underlying portfolios. One can distinguish the following
cases: (i) an index EPS, (ii) a bespoke portfolio EPS, (iii) a cross-currency index EPS, and (iv) a cross-currency bespoke portfolio EPS. A classification is important since only for some instances of an EPS a robust (i.e., model-free) pricing and static hedging can be developed, provided that the associated index and equity options are liquidly traded.

Since EPSs are just a specific class of financial derivatives, they can be analysed using well-established methods from the arbitrage pricing theory based on the concept of continuous-time dynamic trading but then, obviously, the issue of a model choice arises even if the theoretical market completeness is postulated so that the replication argument can be used. In contrast, if the market completeness is not assumed, then one needs to refer to less effective pricing tools developed by academics to deal with incomplete models but it should be acknowledged that these methods are less appreciated by the finance industry. In addition, any stochastic model needs to be calibrated to market data and, in practice, it is also periodically recalibrated in order to mitigate the model risk.

Finally, one can also apply the most recent robust methodology where a wide class of models is simultaneously covered but, once again, a large amount of market data are required to obtain reliable pricing results with a narrow bid-ask spread. To sum up, depending on a class of an EPS and availability of traded options one needs to develop the most effective approach to obtain the fair price interval and hedging strategies.

\subsubsection{Index EPS}    \label{sec2.4.1}

As expected, the easiest to handle is an {\it index} EPS where the process $S$ represents a particular stock index in the domestic economy, such as the S\&P/ASX~200 index, which measures the performance of the 200 largest index-eligible stocks listed on the Australian Securities Exchange (ASX). The related exchange-traded index options are cash-settled calls and puts of European style so they can only be exercised on their expiry date. As will be shown in Section \ref{sec3}, due to the fact that options on S\&P/ASX~200 are actively traded, the problem of pricing and static hedging of a generic EPS referencing S\&P/ASX~200 does not require to use any stochastic model for the index dynamics. Also, it will be argued that the fair value of an index EPS can be easily computed at any date in reference to market data for index options so that an index EPS can be made puttable by the holder at its market value. In principle, they can also be made callable by the provider but such a clause would
not suit the main goal of an EPS, which is to give protection to holders of a superannuation account.

Let $\text{Call}_{\,t}(K,U)$ and $\text{Put}_{\,t}(K,U)$ denote the price at time $t$ of a European call and put option written
on the index $S$, with strike $K \geq 0$ and maturity $U$. For future reference, we observe that the following equalities hold,
for any $ -1 \leq l \leq 0$,
\begin{align*}
\big(l-R_S(0,T)\big)^+=\frac{(S_0(1+l)-S_T)^+}{S_0}=\frac{1}{S_0}\,\text{Put}_T(S_0(1+l),T)
\end{align*}
and, similarly, for any nonnegative $g$,
\begin{align*}
\big(R_S(0,T)-g\big)^+=\frac{(S_T-S_0(1+g))^+}{S_0}=\frac{1}{S_0}\,\text{Call}_T(S_0(1+g),T).
\end{align*}
Notice also that $S_0$ is the initial value of the index and the notional principal of an index EPS is any quantity $N_p>0$ denominated in domestic currency. Hence we may and do take $N_p=1$ since $N_p$ is independent of the value of $S_0$ and our pricing and hedging approach through replication is linear.


\subsubsection{Bespoke Portfolio EPS}    \label{sec2.4.2}

A {\it bespoke portfolio} EPS is referencing an arbitrary static portfolio composed of investments in major companies traded on ASX.
European options on individual stocks are traded on ASX for around 20 top companies. Tailor-made options of both European and American style on a larger selection of Australian companies are offered on an over-the-counter basis by the major banks: ANZ, CBA,
NAB, and Westpac and thus equity options are readily available. However, an EPS on a bespoke portfolio of Australian companies is
related to a basket option and thus its pricing and hedging is necessarily relying on a multidimensional stochastic model for equities,
which will typically require to calibrate volatilities and correlations).

Let $S^i$ represent the market price of the $i$th domestic stock from the reference portfolio so that
$R^i(0,T):=(S^i_T-S^i_0)/S^i_0$ is the realised return on the $i$th domestic stock. Then the portfolio's value process
equals $S_t:=\sum_{i=1}^k n_i S^i_t$ and thus the realised return can be represented as follows
\begin{equation*}
R_S(0,T)=(S_T-S_0)/S_0=\frac{\sum_{i=1}^k n_i S^i_T-\sum_{i=1}^k n_i S^i_0}{\sum_{i=1}^k n_i S^i_0}
=\frac{\sum_{i=1}^k n_i R^i(0,T)S^i_0}{\sum_{i=1}^k N^i_p}
\end{equation*}
where $n_i,\,i=1,2,\dots,k$ is the number of shares of the $i$th stock held at time 0 and thus the wealth invested
at time 0 in the $i$th stock equals $N^i_p:=n_i S^i_0$ and the notional principal equals $N_p =\sum_{i=1}^k N^i_p$.

As in the case of an index EPS, it is convenient to assume, without loss of generality, that the notional principal of an EPS equals
$N_p=1$ so that $\sum_{i=1}^k N^i_p=1$ where $N^i_p>0$ for every $i$. Equivalently, we denote by $\pi^i_0:= n_iS^i_0/\sum_{j=1}^k n_j S^j_0$
a (constant) proportion of the nominal principal invested at time 0 in shares of the $i$th stock so that $\pi^i_0>0$ for every $i$ and $\sum_{i=1}^k \pi^i_0=1$. Under this convention, for the realised return on a bespoke portfolio we obtain
\begin{equation*}
R_S(0,T)=\sum_{i=1}^k N^i_p R^i(0,T)=\sum_{i=1}^k \pi^i_0 R^i(0,T).
\end{equation*}
Consequently, for any number $-1\leq l \leq 0$ we obtain the following equalities
\begin{align*}
(l-R_S(0,T))^+&= \frac{1}{S_0}\,(S_0(1+l)-S_T)^+ =\bigg(l-\sum_{i=1}^k \pi^i_0 R^i(0,T)\bigg)^+
=\bigg((1+l)-\sum_{i=1}^k \pi^i_0 \frac{S^i_T}{S^i_0} \bigg)^+
\end{align*}
where the last equality shows that in the case of a bespoke portfolio we deal with a particular basket put option. Similarly,
for any number $g \geq 0$ we have that
\begin{align*}
(R_S(0,T)-g)^+&=\frac{1}{S_0}\,(S_T-S_0(1+g))^+=\bigg(\sum_{i=1}^k \pi^i_0 R^i(0,T)-g\bigg)^+
=\bigg(\sum_{i=1}^k \pi^i_0 \frac{S^i_T}{S^i_0}-(1+g)\bigg)^+,
\end{align*}
which is the payoff of a basket call option. Since options on a bespoke portfolio are not traded, it can be useful to observe
that the following inequalities hold
\begin{align*}
\bigg(l-\sum_{i=1}^k \pi^i_0 R^i(0,T)\bigg)^+\leq \sum_{i=1}^k \pi^i_0 \big(l-R^i(0,T)\big)^+
=\sum_{i=1}^k \frac{\pi^i_0}{S^i_0}\,\text{Put}_T(S^i_0(1+l),T)
\end{align*}
and
\begin{equation*}
\bigg(\sum_{i=1}^k \pi^i_0 R^i(0,T)-g\bigg)^+\leq \sum_{i=1}^k \pi^i_0 \big(R^i(0,T)-g\big)^+
=\sum_{i=1}^k \frac{\pi^i_0}{S^i_0}\,\text{Call}_T(S^i_0(1+g),T).
\end{equation*}
In general, one can use results pertaining to arbitrage-free pricing of basket options of European style, for instance
in a multidimensional version of the Black-Scholes model.

\subsubsection{Cross-Currency EPS}    \label{sec2.4.3}

A {\it cross-currency market index} EPS references a combination of stock indices in foreign securities markets and possibly also the domestic one. Similarly, a {\it cross-currency bespoke portfolio} EPS is based on the realised return on a bespoke portfolio of domestic and foreign equities. Understandably, the pricing and hedging of a cross-currency EPS is a rather complex exercise but it falls within the well-studied area of arbitrage pricing of cross-currency derivatives. Since, by construction, all payoffs from an EPS are expressed in domestic currency, an analysis of a cross-currency EPS requires stochastic models not only for relevant indices (or stocks) but also applicable exchange rates. It is expected that not only domestic and foreign equity options, but also currency options or forward contracts can be used for hedging of cross-currency EPSs although, typically, a static hedging portfolio is unlikely to be available. For instance, exchange-traded options on U.S. indices and stocks are available for investors in Australia but they are regular options denominated in U.S. dollar and thus the exchange rate risk is not accounted for in their payoff profile, which explains why they are not suitable for construction of a static hedging portfolio for a cross-currency AUD/USD EPS.

For concreteness, we assume that U.S. equities market plays the role of a generic foreign market. Then the exchange rate $Q$ is an $\ff$-adapted stochastic process, which at any date $t$ is quoted as
\begin{equation*}
Q_t=\frac{\text{Number of units of the domestic currency (AUD)}}{\text{One unit of the foreign currency (USD)}}.
\end{equation*}
If $S^i$ represents the market price of the $i$th foreign stock from the reference portfolio, then the realised
return on $S^i$ expressed in domestic currency is given by
\[
R^{Q,i}(0,T):=(Q_T S^i_T- Q_0 S^i_0)/Q_0S^i_0.
\]
Therefore, if a cross-currency EPS references a bespoke portfolio, which is composed of domestic stocks (or indices) $S^i$ for $i=1,2,\dots ,m$ and foreign stocks (or indices) $S^i$ for $i=m+1,m+2,\dots ,k$, then
\[
S_t=\sum_{i=1}^m n_i S^i_t + \sum_{i=m+1}^k n_i Q_t S^i_t
\]
and thus if we assume that $N_p =1$ so that $\sum_{i=1}^k N^i_p=1$ where $N^i_p=n_i S^i_0$ for $i=1,2,\dots ,m$ and $N^i_p=n_i Q_0S^i_0$ for $i=m+1,m+2,\dots ,k$, then
\begin{equation*}
R_S(0,T) 
=\sum_{i=1}^m N^i_p R^i(0,T)+\sum_{i=m+1}^k N^i_p R^{Q,i}(0,T)=\sum_{i=1}^m \pi^i_0 R^i(0,T)+\sum_{i=m+1}^k \pi^i_0 R^{Q,i}(0,T)
\end{equation*}
where, as before, $\pi^i_0$ is a proportion of the nominal principal invested at time 0 in shares of the $i$th stock.

Consequently, for any $l\leq 0$,
\begin{align*}
(l-R_T)^+&= \frac{1}{S_0}\,(S_0(1+l)-S_T)^+=\bigg(l-\sum_{i=1}^m \pi^i_0 R^i(0,T)-\sum_{i=m+1}^k \pi^i_0 R^{Q,i}(0,T)\bigg)^+ \\
& =\bigg((1+l) - \sum_{i=1}^m \pi^i_0 \frac{S^i_T}{S^i_0}-\sum_{i=m+1}^k \pi^i_0 \frac{Q_TS^i_T}{Q_0S^i_0} \bigg)^+
\end{align*}
and, for any $g \geq 0$,
\begin{align*}
(R_T-g)^+&=\frac{1}{S_0}\,(S_T-S_0(1+g))^+=\bigg(\sum_{i=1}^k \pi^i_0 R^i(0,T)+\sum_{i=m+1}^k \pi^i_0 R^{Q,i}(0,T)-g\bigg)^+\\
& =\bigg( \sum_{i=1}^m \pi^i_0 \frac{S^i_T}{S^i_0}-\sum_{i=m+1}^k \pi^i_0 \frac{Q_TS^i_T}{Q_0S^i_0} - (1+g) \bigg)^+.
\end{align*}
It is readily seen that extensions of these expressions to the case of several foreign markets can be obtained. Although we do not elaborate in pricing of hedging of cross-currency EPSs, it should be observed that this can be done in a cross-currency multidimensional variant of the Black-Scholes model complemented by a stochastic model for exchange rates, e.g., the classical Garman-Kohlhagen model.

\section{Hedging Strategies for Equity Protection Swaps}  \label{sec3}

Although an EPS can be based on a virtually arbitrary portfolio of traded domestic and foreign securities, we find it convenient to focus on an EPS referencing a domestic stock index for which there exists a liquid market of related exchange-traded options. The foregoing analysis formally covers also some other EPSs provided that call and put options involved in a static hedging strategy of Proposition \ref{prop5.2.1}
are actively traded either on organised exchanges or OTC market. However, if an EPS involves several domestic indices or a bespoke portfolio of domestic assets, then it represents a multi-asset financial derivative and thus its static hedging using only index options and single-name stock options cannot be achieved.

After representing the payoff of a generic EPS in terms of a combination of European calls and puts, our next goal is to build a hedging portfolio for two particular cases of an EPS and, subsequently, to extend it to a generic EPS. Let us temporarily assume that the provider receives at time 0 from the buyer of an EPS a negative or positive premium $c$ per one unit of the nominal principal. According to our convention, the provider receives at time $T$ the cash flow $\pp (R_T)$ from the buyer.

Since the buyer's goal is to use an EPS as a tool for the reduction of the risk exposure associated with their investment in $S$, so the holder is concerned with the net position $R_T - \pp (R_T)$. In contrast, the provider is not holding the reference portfolio and thus he is only concerned
with elimination of the risk associated with the cash flow $\pp (R_T)$ through replication (if feasible) of the opposite cash flow, that is, the payoff $-\pp (R_T)$ at time $T$. The price of an EPS computed by the provider (seller) of an EPS will also be fair for its holder (buyer), under the assumption that the provider and holder have access to the same set of traded options and the bid-ask spread for traded options is disregarded, e.g., by focusing on the mid-price.

We henceforth take the view of the provider who establishes at time 0 an appropriate portfolio of plain-vanilla call and put options written on the underlying portfolio $S$. We assume that the hedging portfolio, denoted as $H$, costs $H_0$ at time 0 and yields a random payoff $H_T$ at maturity date $T$. Unlike a RILA, the expected expenses of the holder's portfolio are not relevant for the provider since it is not exchanged
at time 0 between the two parties. Hence the management fee charge, which is considered when studying RILAs and other similar variable annuities, is also immaterial because the provider does not manage the holder's portfolio. Note also that if a static hedge of an EPS using traded European options is feasible, then it suffices to independently investigate the hedging portfolios for the two legs of an EPS and combine the respective hedges in order to obtain a complete static hedge for the whole EPS. For simplicity of presentation, we henceforth assume that the continuously compounded annualised risk-free rate equals $r$ although this assumption is not essential.

\begin{definition} \label{EPS CCF} {\rm
For an EPS introduced in Definition \ref{EPS}, the provider's \textit{hedged cash flow} at time $T$ per one unit of the nominal principal $N_p$, denoted by $\CCF_T(c,H)$, equals
\begin{equation*}
\CCF_T(c,H)=(c-H_0)e^{rT}+H_T+\pp (R_T)
\end{equation*}
where the provider's payoff $\pp (R_T)$ from an EPS is given by \eqref{eq0.2.1}. We say that an EPS can be \textit{statically hedged} by the provider if there exists a static portfolio $\whH$ established at time 0 and an initial premium $\whc \in \mathbb{R}$ such the equality $\CCF_T(\whc,\whH)=0$ holds almost surely. Then the number $\whc$ is called the \textit{fair premium} for an EPS per one unit of the nominal principal.}
\end{definition}

\subsection{Static Hedge for a Buffer Index EPS} \label{sec3.1}

In the next step, we establish the existence of a static hedge, and hence the initial premium, for the most typical examples of an index EPS. Furthermore, we analyse particular structures of an EPS for which the fair premium is null and hence there is no cash flow at inception. In Sections \ref{sec3.1}, \ref{sec3.2} and \ref{sec3.3}, we assume that European options written on a reference index $S$ are liquidly traded and we will show that a static hedge for an index EPS is feasible. Therefore, the pricing and hedging results are independent of a stochastic model for the portfolio's value process $S$.

To justify this claim, it suffices to show that the terminal payoff of an EPS can be expressed through a combination of European call and put options written on $S$.  We work here under the postulate that European options with $S$ as the underlying asset are actively traded and we denote by $\text{Call}_{\,t}(K,T)$ and $\text{Put}_{\,t}(K,T)$ denote the price at time $t$ of a European call and put option written on $S$, with strike $K \geq 0$ and maturity $T$.

The provider's cash flow $\CCF_T(c,H)$ for the hedged buffer EPS equals
\begin{align} \label{eq0.7.2}
\CCF_T(c,H)&=(c-H_0)e^{rT}+H_T+ \pp_B (R_T)\nonumber   \\
&=(c-H_0)e^{rT}+H_T - p_2 (l_1-R_T)^+ + f_2 (R_T-g_1)^+ \\
&=(c-H_0)e^{rT}+H_T -\frac{p_2}{S_0}\,\text{Put}_T(S_0(1+l_1),T)+ \frac{f_2}{S_0}\,\text{Call}_T(S_0(1+g_1),T) \nonumber
\end{align}
where to obtain the last equality, it suffices to observe (as in Section \ref{sec2.4.1}) that
\begin{align*}
 (l_1-R_T)^+=\frac{(S_0(1+l_1)-S_T)^+}{S_0}=\frac{1}{S_0}\,\text{Put}_T(S_0(1+l_1),T)
\end{align*}
and, similarly,
\begin{align*}
(R_T-g_1)^+=\frac{(S_T-S_0(1+g_1))^+}{S_0}=\frac{1}{S_0}\,\text{Call}_T(S_0(1+g_1),T).
\end{align*}
From the above equations, it is readily seen that the buffer EPS can be statically hedged by investing in European options
written on $S$. Specifically, we deduce from \eqref{eq0.7.2} that the provider should take the following positions at time $0$:
\begin{itemize}
\item Long $\frac{p_2}{S_0}$ units of a put with strike $K^l_1:=S_0(1+l_1)$,
\item Short $\frac{f_2}{S_0}$ units of a call with strike $K^g_1:=S_0(1+g_1)$,
\end{itemize}
since this portfolio, denoted as $\whH$, provides at time $T$ the payoff
\begin{align*}
\whH_T=\frac{p_2}{S_0}\,\text{Put}_T(K^l_1,T)-\frac{f_2}{S_0}\,\text{Call}_T(K^g_1,T)=p_2 (l_1-R_T)^+ - f_2 (R_T-g_1)^+=- \pp_B (R_T),
\end{align*}
as was required. It is clear that its initial value equals
\begin{align*}
\whH_0= \frac{p_2}{S_0}\,\text{Put}_0(K^l_1,T)-\frac{f_2}{S_0}\,\text{Call}_0(K^g_1,T)
\end{align*}
and for the above choice of the static hedge $\whH$ we have that $\CCF_T(\whc,\whH)=(\whc-\whH_0)e^{rT}=0$, which shows
that the fair premium for a buffer EPS  with the nominal principal $N_p=1$ satisfies $\whc=\whH_0$. Notice also that portfolio's wealth $\whH_T$ can be represented explicitly in terms of the realised value of the index $S_T$  (see Figure \ref{hedge buffer EPS})
\begin{align*}
\whH_T=\frac{p_2}{S_0}\,(K^l_1-S_T)\I_{\{S_T \in [0,K^l_1]\}}-\frac{f_2}{S_0}\,(S_T-K^g_1)\I_{\{S_T \in [K^g_1,\infty[\}}.
\end{align*}

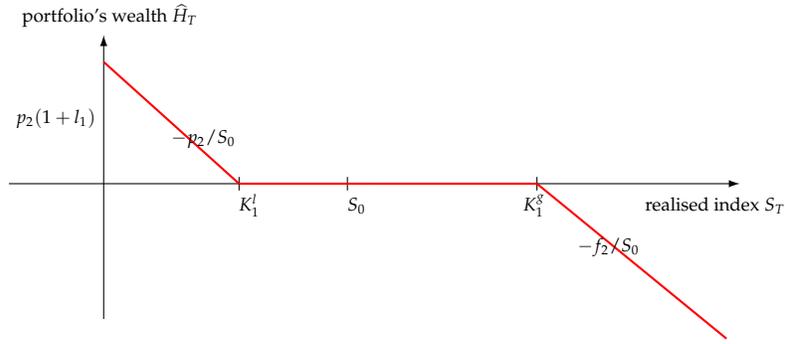
\begin{figure}\centering
\setlength{\unitlength}{1.8mm}
\begin{picture}(24,24)(-10,-10)
\put(0,0){\vector(1,0){27}}
\put(-20,0){\vector(0,1){11}}
\put(0,0){\line(-1,0){27}}
\put(-20,0){\line(0,-1){10}}
\put(12,-0.5){{\line(0,1){1}}}
\put(-10,-0.5){{\line(0,1){1}}}
\put(-2,-0.5){{\line(0,1){1}}}
\thicklines
\put(0,0){{\color{red}\line(1,0){12}}}
\put(12,0){{\color{red}\line(11,-9){14}}}
\put(0,0){{\color{red}\line(-1,0){10}}}
\put(-10,0){{\color{red}\line(-10,9){10}}}
\put(-26,12){\scriptsize portfolio's wealth $\whH_T$}
\put(20,-2){\scriptsize realised index $S_T$}
\put(-26.5,4.5){\scriptsize $p_2(1+l_1)$}
\put(-10,-2){\scriptsize $K_1^l$}
\put(-2,-2){\scriptsize $S_0$}
\put(11,-2){\scriptsize $K_1^g$}
\put(-15,3){\scriptsize $-p_2/S_0$}
\put(15,-5){\scriptsize $-f_2/S_0$}
\end{picture}
\caption{Provider's static hedge for a buffer EPS} \label{hedge buffer EPS}
\end{figure}

\subsection{Static Hedge for a Floor Index EPS}  \label{sec3.2}

An analysis of a static hedge for a floor EPS is analogous to a buffer EPS with some modifications of the payoff.
The provider's cash flow $\CCF_T(c,H)$ for the hedged floor EPS equals, per one unit of the nominal principal,
\begin{align*}
& \CCF_T(c,H)=(c-H_0)e^{rT}+H_T+\pp_F (R_T) \\
&= (c-H_0)e^{rT}+H_T-p_1 (-R_T)^+ +p_1(l_1-R_T)^+ +f_2(R_T-g_1)^+.\nonumber
\end{align*}
We claim that the provider of the floor EPS can statically hedge a floor EPS by investing at time 0 in three European
options on $S$:
\begin{itemize}
\item Short $\frac{p_1}{S_0}$ units of a put with strike $K^l_1:=S_0(1+l_1)$,
\item Long $\frac{p_1}{S_0}$ units of a put with strike $S_0$,
\item Short $\frac{f_2}{S_0}$ units of a call with strike $K^g_1:=S_0(1+g_1)$.
\end{itemize}
This portfolio $\whH$ of European options yields at time $T$ the payoff given by
\begin{align*}
\whH_T& =-\frac{p_1}{S_0}\,\text{Put}_T(K^l_1,T)+\frac{p_1}{S_0}\,\text{Put}_T(S_0,T)- \frac{f_2}{S_0}\,\text{Call}_T(K^g_1,T) \\
&=-p_1(l_1-R_T)^+ + p_1 (-R_T)^+ - f_2 (R_T-g_1)^+ \\
&=-p_1 l_1\I_{\{R_T \in [\lbn,l_1]\}}+p_1(-R_T)^+\I_{\{R_T \in (l_1,0]\}} -f_2(R_T-g_1)^+ = -\pp_F (R_T)
\end{align*}
and thus the fair premium for the floor EPS  with the nominal principal $N_p=1$ equals
\begin{equation*}
\whc=\whH_0=-\frac{p_1}{S_0}\,\text{Put}_0(K^l_1,T)+\frac{p_1}{S_0}\,\text{Put}_0(S_0,T)-\frac{f_2}{S_0}\,\text{Call}_0(K^g_1,T).
\end{equation*}
It can also be checked that $\whH_T$ can be represented as follows (see Figure \ref{hedge floor EPS})
\begin{align*}
\whH_T=\frac{p_1}{S_0}\,(S_0-K^l_1)\I_{\{S_T \in [0,K^l_1]\}}+\frac{p_1}{S_0}(S_0-S_T)\I_{\{S_T \in [K^l_1,S_0]\}}
 -\frac{f_2}{S_0}\,(S_T-K^g_1)\I_{\{S_T \in [K^g_1,\infty[\}}
\end{align*}
where $\frac{p_1}{S_0}\,(S_0-K^l_1)=-p_1 l_1$ is the (positive) difference between the prices at time 0 of put options
with strikes $S_0$ and $K^l_1$.

\begin{figure}\centering
\setlength{\unitlength}{1.8mm}
\begin{picture}(24,24)(-10,-10)
\put(0,0){\vector(1,0){27}}
\put(-20,0){\vector(0,1){11}}
\put(0,0){\line(-1,0){27}}
\put(-20,0){\line(0,-1){10}}
\put(12,-0.5){{\line(0,1){1}}}
\put(-10,-0.5){{\line(0,1){1}}}
\put(-2,-0.5){{\line(0,1){1}}}
\thicklines
\put(0,0){{\color{red}\line(1,0){12}}}      
\put(0,0){{\color{red}\line(-1,0){2}}}      
\put(12,0){{\color{red}\line(11,-9){14}}}   
\put(-20,4){{\color{red}\line(1,0){10}}}    
\put(-10,4){{\color{red}\line(10,-9){8}}}
\put(-26,12){\scriptsize portfolio's wealth $\whH_T$}
\put(20,-2){\scriptsize realised index $S_T$}
\put(-10,-2){\scriptsize $K_1^l$}
\put(-24.5,3.5){\scriptsize $-p_1l_1$}
\put(-2,-2){\scriptsize $S_0$}
\put(11,-2){\scriptsize $K_1^g$}
\put(-6,2){\scriptsize $-p_1/S_0$}
\put(15,-5){\scriptsize $-f_2/S_0$}
\end{picture}
\caption{Provider's static hedge for a floor EPS} \label{hedge floor EPS}
\end{figure}
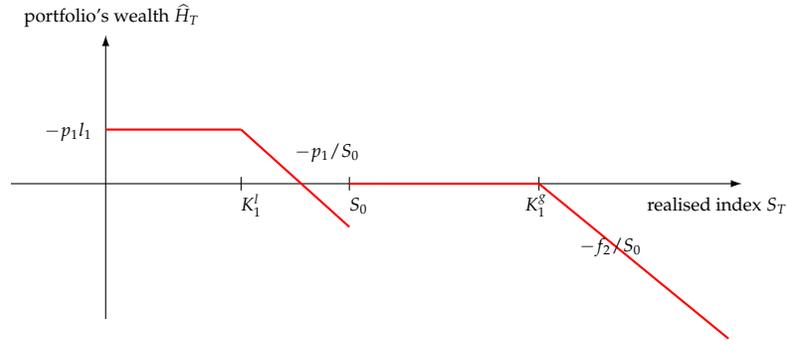

\subsection{Static Hedge for a Generic Index EPS} \label{sec3.3}

We do not present an explicit solution for the provider's static hedge of an EPS with the buffer-floor protection leg and buffer fee leg since it is covered by a result for a generic EPS of Definition~\ref{EPS} (see Example~\ref{buffer_ex1}). Generally speaking, if the parameters $\alpha_k$ and $\beta_k$ in Definition~\ref{EPS} are given, then the question of fair valuation of an EPS reduces to the computation of the fair premium $\whc=\whH _0$. However, a more interesting question arises if it is not assumed a priori that $\whc=0$ but at least one of parameters $\alpha_k$ or $\beta_k$ is left unspecified. Then the issue of fair pricing hinges on finding the value of that parameter (typically, the provider's fee rate) for which the initial value $\whH_0$ is null so that the fair premium $\whc$ vanishes as well. By definition, a {\it static portfolio} in European put and call options is a vector $h=(a_0,a_1,a_2,\dots ,a_n,b_0,b_1,b_2,\dots ,b_m)\in \mathbb{R}^{n+m}$ and the value process $H_t,\, t \in [0,T]$, is given by the equality, for every $t\in [0,T]$,
\begin{equation*}
H_t=\sum_{i=0}^{n}a_i\,\text{\rm Put}_t(K^l_i,T) + \sum_{j=0}^{m}b_j\,\text{\rm Call}_t(K^g_j,T).
\end{equation*}
The above arguments lead to the following hedging result for a generic index EPS with virtually arbitrary structure of the two legs in which
the fair premium is computed in relation to a family of call and put options on $S$ with maturity $T$ and various strikes.
Recall that $l_0=g_0=0$. For convenience, we also adopt the convention that $f_0=p_0=0$.

\begin{proposition} \label{prop5.2.1}
The static hedge $(\whc,\whH)$ composed of call and put options for an index EPS with the nominal principal $N_p=1$ is obtained by setting $a_j=(p_{i+1}-p_i)/S_0$ and $b_j=-(f_{j+1}-f_j)/S_0$ so that
\begin{equation} \label{xeq5.2.1}
\whH_T=\sum_{i=0}^{n}\frac{p_{i+1}-p_i}{S_0}\,\text{\rm Put}_T(K^l_i,T)-\sum_{j=0}^{m}\frac{f_{j+1}-f_j}{S_0}\,\text{\rm Call}_T(K^g_j,T)
 = - \pp (R_T)
\end{equation}
where $K^l_i=S_0(1+l_i)$ and $K^g_j=S_0(1+g_j)$ for every $i=0,1,\dots, n$ and $j=0,1,\dots ,m$ and the fair premium for an EPS
satisfies $\whc = \whH _0$, that is,
\begin{equation}   \label{eq5.2.1}
\whc =\sum_{i=0}^{n}\frac{p_{i+1}-p_i}{S_0}\,\text{\rm Put}_0(K^l_i,T)-\sum_{j=0}^{m}\frac{f_{j+1}-f_j}{S_0}\,\text{\rm Call}_0(K^g_j,T).
\end{equation}
\end{proposition}

\begin{proof}
Observe that, using the ideas introduced in Lemma \ref{lem0.2.1}, we can split an EPS into two parts, the protection and fee legs, and then compute a static hedge for each of them independently. Recall from Lemma \ref{lem0.2.1} that $\pp (R_T)=\pp^f (R_T)-\pp^p (R_T)$ where for every $i=0,1,\dots,n$ and $j=0,1,\dots ,m$
\begin{align*}
\pp^p(R_T)&=\Big(\sum_{k=0}^{i-1}p_{k+1}\Delta l_k+p_{i+1}(R_T-l_i)\Big)\I_{\{R \in (l_{i+1},l_i)\}}, \\
\pp^f(R_T)&=\Big(\sum_{k=0}^{j-1}f_{k+1}\Delta g_k+f_{j+1}(R_T-g_j)\Big)\I_{\{R \in (g_j,g_{j+1})\}}.
\end{align*}
According to Definition \ref{EPS CCF}, it suffices to find a particular static portfolio of options, denoted as $\wh h$, such that $\whH_T + \pp (R_T)=0$ and set $\whc=\whH_0$ so that $\CCF_T(\whc,\whH)=0$.

We first consider the fee leg of an EPS. Observe that either $f_1>0$ or $f_1=0$ and $f_2>0$.
We now start with the first nonzero fee rate $f_{j+1}$ so that $f_{j+1}>f_j=0$ and thus either $f_{j+1}=f_1$ or $f_{j+1}=f_2$.
In order to hedge the fee leg, the provider needs to short $\frac{f_{j+1}}{S_0}$ units of European call option with strike $K^g_j=S_0(1+g_j)$ at time 0. Then, if the inequality $f_{k+1}>f_k$ holds for some $k=j+1,j+2,\dots,m$, the provider needs to short further $\frac{f_{k+1}-f_k}{S_0}$ units of a European call option with strike $K^g_k=S_0(1+g_k)$. If, on the contrary, we have that $f_{k+1}<f_k$, then the provider needs to long $\frac{f_k-f_{k+1}}{S_0}$ units of a European call option with strike $K^g_k=S_0(1+g_k)$. This portfolio of call options constitutes the first part of a static hedge.

It now remains to find a hedge for the protection leg of an EPS. Notice that we have that either $p_1>0$ or $p_1=0$ and $p_2>0$.
We start with the first nonzero protection rate $p_{i+1}$ so that $p_{i+1}>p_i=0$ where, by convention, $p_0=0$ so that either $p_{i+1}=p_1$ or $p_{i+1}=p_2$. In order to hedge the protection leg of an EPS, the provider needs to long $\frac{p_{i+1}}{S_0}$ units of European put option with strike $K^l_i=S_0(1+l_i)$ at the inception date of an EPS. Then, for any $k$ such that $p_{k+1}>p_k$ with $k=i+1,i+2,\dots,n$, the provider should long $\frac{p_{k+1}-p_k}{S_0}$ units of a European put option with strike $K^l_k=S_0(1+l_k)$. If, however, the inequality $p_{k+1}<p_k$ holds, then the provider need to short $\frac{p_k-p_{k+1}}{S_0}$ units of a European put option with strike $K^l_k=S_0(1+l_k)$. These put options give the static hedge for the fee leg.

The equality $\whc = \whH _0$ and hence equality \eqref{eq5.2.1} for the fair premium $\whc$ are now immediate consequences of Definition \ref{EPS CCF}.
\end{proof}

In Proposition \ref{prop5.2.1} we offer a solution to the hedging problem for an EPS but the issue of uniqueness of a static hedge for an EPS is not examined and, as usual when dealing with a combination of European options, it may fail to hold, in general. However, the uniqueness of the fair premium for an EPS relative to a set of market prices of traded options is always true since its non-uniqueness would be inconsistent with the principle of arbitrage-free pricing. In the context of Proposition \ref{prop5.2.1}, European options written on $S$ maturing at $T$ are assumed to play the role of primary traded assets and an EPS is a derivative security. Hence the pricing formula \eqref{eq5.2.1} is {\it robust,} meaning that it does not depend on a choice of a stochastic model for $S$, provided that options prices are observed.
However, Proposition \ref{prop5.2.1} is also valid in the framework of any arbitrage-free model for assets $B$ and $S$ where $B$ represents the money market account and $S$ is assumed to be a primary traded asset so that European options are derivative securities and their arbitrage-free prices are computed, not observed on the market.
Finally, it is clear the same arguments can be applied the case of an EPS written on a bespoke portfolio of domestic and foreign indices or stocks although, for technical reasons, the problem may become more difficult to solve.

\begin{example} \label{buffer_ex1} {\rm
 As a practical illustration of the static hedge $(\whc,\whH)$ obtained in Proposition \ref{prop5.2.1}, let us consider an index EPS with buffer-floor protection leg and buffer-cap fee leg, as introduced in Definitions \ref{Fee} and \ref{Prot}. Then the modified
realised return equals $\pp (R)=\pp^{p,4}(R)-\pp^{f,4}(R)$ where
\begin{align*}
\pp^{p,4}(R)=p_2(R-l_1)\I_{\{R \in [l_2,l_1)\}}+p_2(l_2-l_1)\I_{\{R \in [\lbn,l_2)\}}
\end{align*}
and
\begin{align*}
\pp^{f,4}(R)=f_2(R-g_1)\I_{\{R \in (g_1,g_2]\}}+f_2(g_2-g_1)\I_{\{R \in (g_2,\infty)\}}.
\end{align*}
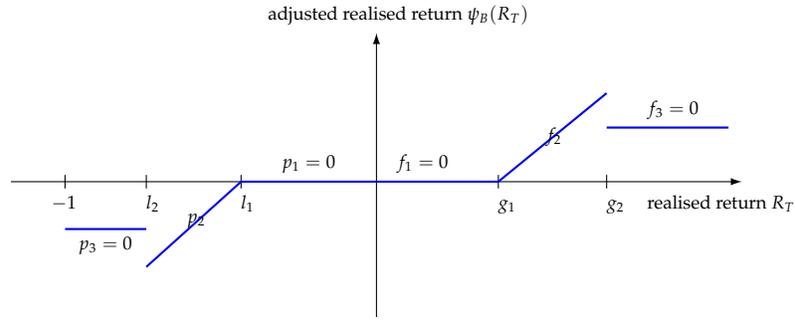
\begin{figure}\centering
\setlength{\unitlength}{1.8mm}
\begin{picture}(24,24)(-10,-10)
\put(0,0){\vector(1,0){27}}
\put(0,0){\vector(0,1){11}}
\put(0,0){\line(-1,0){27}}
\put(0,0){\line(0,-1){10}}
\put(9,-0.5){{\line(0,1){1}}}
\put(17,-0.5){{\line(0,1){1}}}
\put(-10,-0.5){{\line(0,1){1}}}
\put(-17,-0.5){{\line(0,1){1}}}
\put(-23,-0.5){{\line(0,1){1}}}
\thicklines
\put(0,0){{\color{blue}\line(1,0){9}}}
\put(9,0){{\color{blue}\line(11,9){8}}}
\put(17,4){{\color{blue}\line(1,0){9}}}
\put(0,0){{\color{blue}\line(-1,0){10}}}
\put(-10,0){{\color{blue}\line(-10,-9){7}}}
\put(-17,-3.5){{\color{blue}\line(-1,0){6}}}
\put(-8,12){\scriptsize adjusted realised return $\pp_B(R_T)$}
\put(20,-2){\scriptsize realised return $R_T$}
\put(-24,-2){\scriptsize $-1$}
\put(-10,-2){\scriptsize $l_1$}
\put(-17,-2){\scriptsize $l_2$}
\put(9,-2){\scriptsize $g_1$}
\put(17,-2){\scriptsize $g_2$}
\put(-7,1){\scriptsize $p_1=0$}
\put(-14,-3){\scriptsize $p_2$}
\put(-22,-5){\scriptsize $p_3=0$}
\put(1.5,1){\scriptsize $f_1=0$}
\put(12.5,3){\scriptsize $f_2$}
\put(20,5){\scriptsize $f_3=0$}
\end{picture}
\caption{Provider's adjusted return for a buffer-floor/buffer-cap EPS} \label{buffer EPS1}
\end{figure}
Since $f_1=f_3=0, f_2\in (0,1]$ and $p_1=p_3=0, p_2\in (0,1]$ a direct application of equality \eqref{xeq5.2.1} with $n=m=2$ yields
\begin{align*}
\whH_T
&=\frac{p_{2}}{S_0}\,\Big(\text{\rm Put}_T(K^l_1,T)
-\text{\rm Put}_T(K^l_2,T)\Big) -\frac{f_{2}}{S_0}\Big(\text{\rm Call}_T(K^g_1,T)-\text{\rm Call}_T(K^g_2,T)\Big)
\end{align*}
where $K^l_i=S_0(1+l_i)<S_0$ and $K^l_i=S_0(1+g_i)>S_0$ for $i=1,2$ and, obviously, $\whc=\whH_0$.
\begin{figure}\centering
\setlength{\unitlength}{1.8mm}
\begin{picture}(24,24)(-10,-10)
\put(0,0){\vector(1,0){27}}
\put(-20,0){\vector(0,1){11}}
\put(0,0){\line(-1,0){27}}
\put(-20,0){\line(0,-1){10}}
\put(10,-0.5){{\line(0,1){1}}}
\put(17,-0.5){{\line(0,1){1}}}
\put(3,-0.5){{\line(0,1){1}}}
\put(-12,-0.5){{\line(0,1){1}}}
\put(-4,-0.5){{\line(0,1){1}}}
\thicklines
\put(-2,0){{\color{red}\line(1,0){12}}}      
\put(-2,0){{\color{red}\line(-1,0){2}}}      
\put(10,0){{\color{red}\line(11,-9){7}}}    
\put(-20,4){{\color{red}\line(1,0){8}}}      
\put(17,-3.5){{\color{red}\line(1,0){9}}}      
\put(-12,4){{\color{red}\line(10,-9){8}}}
\put(-26,12){\scriptsize portfolio's wealth $\whH_T$}
\put(20,-2){\scriptsize realised index $S_T$}
\put(-12,-2){\scriptsize $K^l_2$}
\put(-27,3.5){\scriptsize $p_2(l_1-l_2)$}
\put(-4,-2){\scriptsize $K^l_1$}
\put(2,-2){\scriptsize $S_0$}
\put(9,-2){\scriptsize $K_1^g$}
\put(17,-2){\scriptsize $K_2^g$}
\put(-8.5,2.5){\scriptsize $-p_2/S_0$}
\put(11,-4){\scriptsize $-f_2/S_0$}
\end{picture}
\caption{Provider's static hedge for a buffer-floor/buffer-cap EPS} \label{hedge floor buffer EPS}
\end{figure}
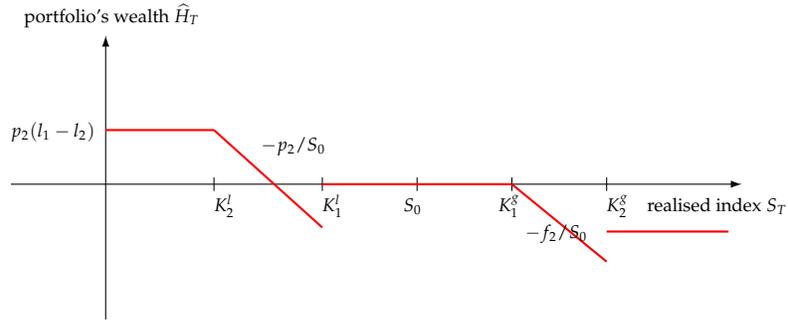

} \end{example}

\section{Forward Performance Testing of Buffer and Floor EPSs} \label{sec4}

To provide the reader with a preliminary insight into the properties of an EPS, we first report results of forward performance tests for some classes of EPSs. For simplicity and concreteness, we employ the Black-Scholes options pricing formula to compute the fair premium for an EPS, as specified through its static hedging strategy of Definition \ref{EPS CCF}.

As in the classical Black-Scholes model, we assume that the dynamics of the process $S$ under the martingale probability measure $\Q$ are given by the SDE
\begin{equation} \label{BSM}
dS_t = (r-\kappa) S_t\,dt + \sigma S_t\,dW_t , \quad S_0 >0 ,
\end{equation}
where $W$ is a one-dimensional standard Brownian motion, $\sigma >0$ is a constant volatility parameter, $r$ is a constant risk-free interest rate,
and $\kappa$ is a constant dividend yield. Obviously, this also means that the process $Z_t=S_t/S_0$ satisfies
\begin{equation} \label{BSMR}
dZ_t = (r-\kappa) Z_t\,dt + \sigma Z_t\,dW_t , \quad Z_0=1 .
\end{equation}
Under the postulated dynamics of the portfolio's value process, we have the following well-known result from \cite{BS1973}, which gives the arbitrage-free prices of European call and put options on $S$ with a fixed strike $K>0$ and maturity date $T$.

\newpage

\begin{theorem} \label{BSE}
The arbitrage-free price of a European call option with the payoff $(S_T-K)^+$ at maturity date $T$ equals, for every $t\in [0,T)$,
\begin{equation*} 
C(S_t,t)=e^{-\kappa(T-t)}S_t N\big(d_+(S_t,T-t)\big) - e^{-r(T-t)}K N\big(d_-(S_t,T-t)\big)
\end{equation*}
where
\begin{equation*}
d_{\pm}(S_t,T-t) = \frac{1}{\sigma \sqrt{T-t}} \bigg[\ln \frac{S_t}{K} + \bigg(r - \kappa \pm \frac{1}{2} \sigma^2\bigg)(T-t) \bigg]
\end{equation*}
and $N$ is the cumulative density function of the standard normal distribution. Similarly, the arbitrage-free price of a European put option
with the payoff $(K-S_T)^+$ at maturity date $T$ is given by, for every $t\in [0,T)$,
\begin{equation*} 
P(S_t,t)= e^{-r(T-t)} K N\big(-d_-(S_t,T-t)\big) -e^{-\kappa (T-t)} S_t N\big(-d_+(S_t,T-t)\big).
\end{equation*}
\end{theorem}

By combining Proposition \ref{prop5.2.1}  with Theorem \ref{BSE}, we immediately obtain the following result.

\begin{proposition} \label{BSEc}
The arbitrage-free price at time 0 in the Black-Scholes model of a generic index EPS with the nominal principal $N_p=1$ equals
\begin{align*}
\whc &=\sum_{i=0}^{n}(p_{i+1}-p_i)\Big(e^{-rT} (1+l_i) N\big(-h_-(l_i,T)\big)-e^{-\kappa T} N\big(-h_+(l_i,T)\big) \Big) \\
& -\sum_{j=0}^{m}(f_{j+1}-f_j)\,\Big(e^{-\kappa T} N\big(h_+(g_j,T)\big)-e^{-rT}(1+g_j)N\big(h_-(g_j,T)\big)\Big)
\end{align*}
where we denote, for every $x>-1$,
\begin{equation*}
h_{\pm}(x,T) = \frac{1}{\sigma \sqrt{T}} \bigg[-\ln (1+x)+ \bigg(r - \kappa \pm \frac{1}{2} \sigma^2\bigg)T\bigg].
\end{equation*}
\end{proposition}

Notice that the price $\whc$ of an EPS does not depend on the initial value $S_0$ of the underlying portfolio (e.g., a reference index). This remarkable property was expected since the arbitrage-free price of an EPS should depend on the interest rate $r$, the dividend yield $\kappa$, the notional principal $N_p$ (here chosen to be equal to 1) and the volatility $\sigma$ of the random return rate $R_T$ on portfolio $S$, but it should be independent on the actual size of a reference portfolio $S$. This is due to the fact that the future returns on investor's holdings in a portfolio $S$ coincide with the future returns on the reference portfolio $S$ (e.g., the S\&P~500 index). Of course, to obtain the pricing formula for an EPS, the dynamics of the process $S$ was specified under the martingale probability measure $\Q$ but the real-world benefits from holding an EPS will be examined in Section \ref{sec5} through backtesting under the statistical probability $\mathbb{P}$.

We stress that the Black-Scholes prices of European options are only used here for illustrative purposes. It is well known that the Black-Scholes model has some limitations and it is not fully consistent with the real-world options market. In Section \ref{sec4}, we use this model in order to compare the theoretical fair premia for EPSs with different parameters but in Section \ref{sec5} we refer instead to the market data for traded European options to obtain empirical pricing results and conclusions in regard of the price dependence on a structure of an EPS. We also give a preliminary analysis of sensitivities of the price of an EPS with respect to the underlying asset $S$ and model parameters.

Needless to say, the pricing of an EPS is dynamic and thus the arbitrage-free price of a fair EPS at any date after its inception date is usually nonzero and with random values that can be either positive or negative for each party. It is thus natural to ask the following question: suppose that a holder decided to enter into a one-year EPS with inception date 0: what should be the arbitrage-free price at time 0 of an option to extend the EPS with the same covenants for the second year at no additional cost at time $T=1$? More generally, one can examine an option to extend an EPS at time $T=1$ for the next $n$ years at no additional cost at time $T=1$. It is rather clear that the answer to the question of fair valuation of options to extend an EPS hinges on arbitrage-free pricing of a particular compound option on $S$ in either the domestic market or a cross-currency framework so the usual pricing techniques based on dynamic hedging can be applied.

\subsection{Theoretical Black-Scholes Pricing of an EPS} \label{sec4.1}

We first compute the fair premia for typical instances of an EPS from the perspective of their provider whose aim is to
offer an index EPS to members of super funds. As was explained, in practice the pricing of an EPS should rely on observed
market prices of traded options but, for clarity of presentation, we assume in Section \ref{sec4} that arbitrage-free option prices
are obtained using Theorem \ref{BSE}. It is well known that the Black-Scholes model does not yield option prices that perfectly match the market data across various strikes and maturities. Nevertheless, it is widely accepted by practitioners for market quotes for options via their implied volatilities and is largely sufficient for illustration of the model-free pricing method established in Proposition \ref{prop5.2.1}. Regarding the parameters of the Black-Scholes dynamics, we assume throughout that the interest rate $r=1.5\%$, the dividend yield $\kappa=0$ (consistently with the market convention for stock indices), and the volatility parameter $\sigma=20\%$ (as a proxy for the one-year volatility index).

\begin{table}[h!]
\centering
\caption{Theoretical fair premia of buffer and floor equity protection swaps}
\begin{tabular}{cccccccccc}
\hline
No. &EPS class &$T$&$l_1$&$g_1$&$p_1$&$p_2$&$f_1$&$f_2$&Fair premium \\
\hline
[1]&Buffer & 1 year & -5\% & 5\% & 0 & 0.5 & 0 & 0.5 & -0.007890 \\

[2]&Buffer & 1 year & -5\% & 5\% & 0 & 0.8 & 0 & 0.5 & 0.006873 \\

[3]&Buffer & 1 year & -5\% & 5\% & 0 & 0.8 & 0 & 0.8 & -0.012623 \\

[4]&Buffer & 1 year & -5\% & 10\% & 0 & 0.5 & 0 & 0.5 & 0.000730 \\

[5]&Buffer & 1 year & -5\% & 10\% & 0 & 0.6 & 0 & 0.5 & 0.005651 \\

[6]&Buffer & 1 year & -5\% & 10\% & 0 & 0.7 & 0 & 0.5 & 0.010572 \\

[7]&Buffer & 1 year & -5\% & 10\% & 0 & 0.8 & 0 & 0.5 & 0.015493 \\

[8]&Buffer & 1 year & -5\% & 10\% & 0 & 0.9 & 0 & 0.5 & 0.020414 \\

[9]&Buffer & 1 year & -5\% & 10\% & 0 & 0.8 & 0 & 0.8 & 0.001168 \\

[10]&Buffer & 1 year & -10\% & 10\% & 0 & 0.5 & 0 & 0.5 & -0.008082 \\

[11]&Buffer & 1 year & -10\% & 10\% & 0 & 0.8 & 0 & 0.5 & 0.001393 \\

[12]&Buffer & 1 year & -10\% & 10\% & 0 & 0.8 & 0 & 0.8 & -0.012931 \\

[13]&Buffer & 2 year & -5\% & 5\% & 0 & 0.8 & 0 & 0.5 & 0.006601 \\

[14]&Buffer & 2 year & -5\% & 10\% & 0 & 0.8 & 0 & 0.5 & 0.015960 \\

[15]&Buffer & 2 year & -10\% & 10\% & 0 & 0.8 & 0 & 0.5 & 0.000237 \\

\hline
[1]&Floor & 1 year & -5\% & 5\% & 0.5 & 0 & 0 & 0.5 & -0.021180 \\

[2]&Floor & 1 year & -5\% & 5\% & 0.8 & 0 & 0 & 0.5 & -0.014391 \\

[3]&Floor & 1 year & -5\% & 10\% & 0.5 & 0 & 0 & 0.5 & -0.012560 \\

[4]&Floor & 1 year & -5\% & 10\% & 0.8 & 0 & 0 & 0.5 & -0.005771 \\

[5]&Floor & 1 year & -5\% & 10\% & 0.8 & 0 & 0 & 0.8 & -0.049895 \\

[6]&Floor & 1 year & -10\% & 10\% & 0.5 & 0 & 0 & 0.5 & -0.003748 \\

[7]&Floor & 1 year & -10\% & 10\% & 0.8 & 0 & 0 & 0.5 & 0.008328 \\

[8]&Floor & 1 year & -15\% & 10\% & 0.5 & 0 & 0 & 0.5 & 0.002672 \\

[9]&Floor & 1 year & -15\% & 10\% & 0.6 & 0 & 0 & 0.5 & 0.007982 \\

[10]&Floor & 1 year & -15\% & 10\% & 0.7 & 0 & 0 & 0.5 & 0.013291 \\

[11]&Floor & 1 year & -15\% & 10\% & 0.8 & 0 & 0 & 0.5 & 0.018601 \\

[12]&Floor & 1 year & -15\% & 10\% & 0.9 & 0 & 0 & 0.5 & 0.023910 \\

[13]&Floor & 1 year & -15\% & 10\% & 0.8 & 0 & 0 & 0.8 & 0.004276 \\

[14]&Floor & 2 year & -5\% & 5\% & 0.8 & 0 & 0 & 0.5 & -0.033581 \\

[15]&Floor & 2 year & -5\% & 10\% & 0.8 & 0 & 0 & 0.5 & -0.024222 \\

[16]&Floor & 2 year & -10\% & 10\% & 0.8 & 0 & 0 & 0.5 & -0.008500 \\

[17]&Floor & 2 year & -15\% & 10\% & 0.8 & 0 & 0 & 0.5 & 0.004358 \\

[18]&Floor & 2 year & -15\% & 10\% & 0.8 & 0 & 0 & 0.8 & -0.021315 \\
\hline
\end{tabular}
\label{table:EPS}
\end{table}

In Table \ref{table:EPS}, we report the fair premia $\whc$ for selected buffer and floor EPSs. Recall from Proposition \ref{prop5.2.1} that the fair premium satisfies $\whc=\whH_0$ and thus its value can be found using the Black-Scholes formula, as was shown in Proposition \ref{BSEc}. We give the fair premia for fifteen buffer EPSs and eighteen floor EPSs with different sets of parameters.

\newpage

Let us analyse some basic features of hedging costs for an EPS. First, we can conclude that for a buffer EPS, the maturity of either one or two years does not make a big difference on the fair premium since its static hedge involves only two options with opposite directions (a long put and a short call). In contrast, the fair premium for a floor EPS is markedly different when the maturity increases. Thus for an EPS with same types of options and different directions, the maturity should not make a big influence on fair premium with all else variables equal. We need to mention that, similar to RILA examined in \cite{M2021}, we find it natural to assume that one-year or two-year EPSs should be offered by a typical provider and some providers may also offer contracts with longer maturities.

Recall from Definition \ref{BEPS} of the buffer EPS that $l_1$ is the level of negative returns that triggers protection and $g_1$ is the level of positive returns such that the provider begins to collect fees. From rows 2, 7, 11 in the top half of Table \ref{table:EPS} we can see that, with other variables being fixed, if either $l_1$ or $g_1$ increases, then the fair premium grows. According to Definition \ref{FEPS},  the parameter $l_1$ in the floor EPS is the level of negative returns such that the provider does not offer any additional compensation beyond that level. As can be seen from cases 4, 7, 11 in the bottom half of Table \ref{table:EPS} the fair premium for the floor EPS increases when the level $l_1$ decreases, as was expected.

%
%
%
%
%
%

Recall that the parameters $p_2$ and $f_2$ of the buffer EPS represent the provider's participation rates in portfolio's losses and gains, respectively. Obviously, the EPS buyer would appreciate if the provider bears more losses and is entitled to less gains. This explains why the buyer is required to pay a higher fair premium when the protection rate $p_2$ is increased, as can be seen from the top half of Table \ref{subtable2:EPS}. Similarly, the buyer still should pay a higher fair premium when the fee rate $f_2$ is lower, as shown in rows 2, 3, 7, 9, 11, and 12 in the buffer EPS of Table \ref{table:EPS}.

Analogous properties can be observed for the floor EPS,  The fair premium is higher when the protection rate $p_1$ is increased, as can be seen in rows 8-12 in the bottom half of Table \ref{subtable3:EPS}, and is lower when the fee rate $f_2$ is increased, as can be observed in rows 11, 13, 17, 18 in the bottom half of Table \ref{subtable3:EPS}. We stress that all these properties (except for the one regarding maturity) can be formally established by considering the dependence of the payoff profile of an EPS on its parameters and hence they are universally true.

\begin{table}[h!]
\centering
\caption{Fair premium dependence on parameters $p_1$ (floor) and $p_2$ (buffer)}
\begin{tabular}{cccccccccc}
\hline
No. &EPS class &$T$&$l_1$&$g_1$&$p_1$&$p_2$&$f_1$&$f_2$&Fair premium \\
\hline
[4]&Buffer & 1 year & -5\% & 10\% & 0 & 0.5 & 0 & 0.5 & 0.000730 \\

[5]&Buffer & 1 year & -5\% & 10\% & 0 & 0.6 & 0 & 0.5 & 0.005651 \\

[6]&Buffer & 1 year & -5\% & 10\% & 0 & 0.7 & 0 & 0.5 & 0.010572 \\

[7]&Buffer & 1 year & -5\% & 10\% & 0 & 0.8 & 0 & 0.5 & 0.015493 \\

[8]&Buffer & 1 year & -5\% & 10\% & 0 & 0.9 & 0 & 0.5 & 0.020414 \\

\hline
[8]&Floor & 1 year & -15\% & 10\% & 0.5 & 0 & 0 & 0.5 & 0.002672 \\

[9]&Floor & 1 year & -15\% & 10\% & 0.6 & 0 & 0 & 0.5 & 0.007982 \\

[10]&Floor & 1 year & -15\% & 10\% & 0.7 & 0 & 0 & 0.5 & 0.013291 \\

[11]&Floor & 1 year & -15\% & 10\% & 0.8 & 0 & 0 & 0.5 & 0.018601 \\

[12]&Floor & 1 year & -15\% & 10\% & 0.9 & 0 & 0 & 0.5 & 0.023910 \\

\hline
\end{tabular}
\label{subtable2:EPS}
\end{table}

\begin{table}[h!]
\centering
\caption{Fair premium dependence on fee rate $f_2$}
\begin{tabular}{cccccccccc}
\hline
No. &EPS class &$T$&$l_1$&$g_1$&$p_1$&$p_2$&$f_1$&$f_2$&Fair premium \\
\hline
[2]&Buffer & 1 year & -5\% & 5\% & 0 & 0.8 & 0 & 0.5 & 0.006873 \\

[3]&Buffer & 1 year & -5\% & 5\% & 0 & 0.8 & 0 & 0.8 & -0.012623 \\
\hline

[7]&Buffer & 1 year & -5\% & 10\% & 0 & 0.8 & 0 & 0.5 & 0.015493 \\

[9]&Buffer & 1 year & -5\% & 10\% & 0 & 0.8 & 0 & 0.8 & 0.001168 \\
\hline

[11]&Buffer & 1 year & -10\% & 10\% & 0 & 0.8 & 0 & 0.5 & 0.001393 \\

[12]&Buffer & 1 year & -10\% & 10\% & 0 & 0.8 & 0 & 0.8 & -0.012931 \\

\hline
[11]&Floor & 1 year & -15\% & 10\% & 0.8 & 0 & 0 & 0.5 & 0.018601 \\

[13]&Floor & 1 year & -15\% & 10\% & 0.8 & 0 & 0 & 0.8 & 0.004276 \\

\hline
[17]&Floor & 2 year & -15\% & 10\% & 0.8 & 0 & 0 & 0.5 & 0.004358 \\

[18]&Floor & 2 year & -15\% & 10\% & 0.8 & 0 & 0 & 0.8 & -0.021315 \\
\hline
\end{tabular}
\label{subtable3:EPS}
\end{table}

When comparing buffer EPSs with floor EPSs, we observe that the hedging costs for buffer EPSs are higher than for floor EPSs. In order to explain this feature, we focus on their respective structures and we note that the fee legs for the buffer EPS and floor EPS shown in Table \ref{table:EPS} are identical but their protection legs differ. From the profile of cash flow function for a buffer EPS (Figure \ref{buffer EPS}) and a floor EPS (Figure \ref{floor EPS}), we can see that in a floor EPS the protection payment for large losses is capped, while it is unlimited in a buffer EPS. Hence the protection provided by a buffer EPS is more effective and, consequently, the hedging cost (i.e., the fair premium) for a buffer EPS is typically higher than for a floor EPS with identical fee leg.

\subsection{Numerical Study of a Fair EPS} \label{sec4.2}

In order to make EPS products more attractive to holders of superannuation accounts, we propose to set the fair premium to zero so that the buyer of an EPS is not required to pay an initial premium to the provider. Furthermore, since investors usually pay more attention to limit their losses, rather than to maximise their gains, we will set the protection rate together with the null fair premium, and then attempt to find a proper fee rate $f_2$. In Table \ref{table:EPS_zeroc}, we give the fee rate $f_2$ when all other parameters are fixed and the fair premium is null. We continue studying some examples of EPSs from the Table \ref{table:EPS_zeroc}.

\begin{table}[h!]
\centering
\caption{Buffer and floor EPSs with null fair premium}
\begin{tabular}{ccccccccc}
\hline
No. &EPS class &$T$&$l_1$&$g_1$&$p_1$&$p_2$&$f_1$&$f_2$ \\
\hline
[1]&Buffer & 1 year & -5\% & 5\% & 0 & 0.5 & 0 & 0.38 \\

[2]&Buffer & 1 year & -5\% & 5\% & 0 & 0.8 & 0 & 0.61 \\

[3]&Buffer & 1 year & -5\% & 10\% & 0 & 0.5 & 0 & 0.52 \\

[4]&Buffer & 1 year & -5\% & 10\% & 0 & 0.6 & 0 & 0.62 \\

[5]&Buffer & 1 year & -5\% & 10\% & 0 & 0.7 & 0 & 0.72 \\

[6]&Buffer & 1 year & -5\% & 10\% & 0 & 0.8 & 0 & 0.82 \\

[7]&Buffer & 1 year & -5\% & 10\% & 0 & 0.9 & 0 & 0.93 \\

[8]&Buffer & 1 year & -10\% & 10\% & 0 & 0.5 & 0 & 0.33 \\

[10]&Buffer & 1 year & -10\% & 10\% & 0 & 0.8 & 0 & 0.53 \\

[11]&Buffer & 2 year & -5\% & 5\% & 0 & 0.8 & 0 & 0.56 \\

[12]&Buffer & 2 year & -5\% & 10\% & 0 & 0.8 & 0 & 0.69 \\

[13]&Buffer & 2 year & -10\% & 10\% & 0 & 0.8 & 0 & 0.50 \\

\hline
[1]&Floor & 1 year & -5\% & 5\% & 0.5 & 0 & 0 & 0.17 \\

[2]&Floor & 1 year & -5\% & 5\% & 0.8 & 0 & 0 & 0.28 \\

[3]&Floor & 1 year & -5\% & 10\% & 0.5 & 0 & 0 & 0.24 \\

[4]&Floor & 1 year & -5\% & 10\% & 0.8 & 0 & 0 & 0.38 \\

[5]&Floor & 1 year & -10\% & 10\% & 0.5 & 0 & 0 & 0.42 \\

[6]&Floor & 1 year & -10\% & 10\% & 0.8 & 0 & 0 & 0.67\\

[7]&Floor & 1 year & -15\% & 10\% & 0.5 & 0 & 0 & 0.56 \\

[8]&Floor & 1 year & -15\% & 10\% & 0.6 & 0 & 0 & 0.67 \\

[9]&Floor & 1 year & -15\% & 10\% & 0.7 & 0 & 0 & 0.78 \\

[10]&Floor & 1 year & -15\% & 10\% & 0.8 & 0 & 0 & 0.89 \\

[11]&Floor & 1 year & -15\% & 10\% & 0.9 & 0 & 0 & 1.0 \\

[12]&Floor & 2 year & -5\% & 5\% & 0.8 & 0 & 0 & 0.18 \\

[13]&Floor & 2 year & -5\% & 10\% & 0.8 & 0 & 0 & 0.22 \\

[14]&Floor & 2 year & -10\% & 10\% & 0.8 & 0 & 0 & 0.40 \\

[15]&Floor & 2 year & -15\% & 10\% & 0.8 & 0 & 0 & 0.55 \\

\hline
\end{tabular}
\label{table:EPS_zeroc}
\end{table}

\begin{example}
\label{bufex1} {\rm
We take the buffer EPS No. 6 from Table \ref{table:EPS_zeroc} as an example. Assume that a holder uses \$1 mm as the nominal principal, the EPS product has the buffer threshold $l_1=-5\%$ with the protection rate $p_2=80\%$, the cap threshold $g_1=10\%$, and maturity one year. The fair premium for this product is null, which means that there no cost for EPS buyer at the inception of the EPS, if the fee rate equals $f_2=82\%$.}
\begin{itemize}{\rm
\item If the realised rate of return $R_T$ is above $10\%$, say $12\%$, then the EPS buyer pays $(12\% - 10\%) \times 82\% \times  \$1\ \mbox{\rm mm}  = \$16.4$ k
\item If $R_T$ is between $-5\%$ and $10\%$, there is no cash flow between the buyer and provider.
\item If $R_T$ is below $-5\%$, say $-8\%$, then the buyer receives $(-8\% + 5\%) \times 80\% \times  \$1\ \mbox{\rm mm}  = \$24$ k.
}\end{itemize}
\end{example}

If a high loss arises due to a negative rate of realised return $R_T$, the provider offers the same level of protection to a holder in both Examples \ref{buffer_ex} and \ref{bufex1}. In contrast, the EPS buyer is required to share more gains from positive returns with the provider in fair EPS. However, we contend that a typical holder would pay more attention to the protection level from potential losses than the retention rate of potential gains and hence a fair EPS would be preferred by most holders. The the protection rate is the same in both EPSs but the EPS from Example \ref{bufex1} has null fair premium and hence the latter EPS is likely to have a greater appeal to investors.

\begin{example} \label{floorex2} {\rm
We now examine the floor EPS No. 10 from Table \ref{table:EPS_zeroc} and we assume that the nominal principal equals \$1 mm. The EPS has the buffer threshold $l_1=-15\%$ with the protection rate $p_1=80\%$, the cap threshold $g_1=10\%$ with the fee rate $f_2=89\%$, and one year maturity. Then the holder does not need to pay an initial premium to the provider, since the fair premium is null.}
\begin{itemize} {\rm
\item If the realised rate of return $R_T$ is above $10\%$, say $12\%$, then the buyer pays $(12\% - 10\%) \times 89\% \times
 \$1\ \mbox{\rm mm}  = \$19.8$ k to the EPS provider.
\item If $R_T$ is between $0\%$ and $10\%$, there is no cash flow between the buyer and provider.
\item If $R_T$ is between $-15\%$ and $0\%$, such that $-8\%$, then the provider pays $8\% \times 80\% \times  \$1\ \mbox{\rm mm}  = \$64$ k
to the buyer.
\item If $R_T$ is below $-15\%$, say $-20\%$, then the provider pays $15\% \times 80\% \times  \$1\ \mbox{\rm mm}  = \$120$ k
 to the buyer. This is the fixed amount that the buyer receives when $R_T$ is less than $-15\%$.
} \end{itemize}
\end{example}

From Table \ref{table:EPS_zeroc}, we can conclude that if the protection rate ($p_2$ for the buffer EPS and $p_1$ for the floor EPS) increases, then the proper level of fee rate $f_2$ increases as well. Rows 3-7 for the buffer EPS in the top half of Table \ref{table:EPS_zeroc} and rows 7-11 in the bottom half of Table \ref{table:EPS_zeroc} justify this observation. For the reader's convenience, these EPSs are also presented in Table \ref{subtable1:EPS_zeroc}.

\begin{table}[h!]
\centering
\caption{Buffer and floor EPSs with null fair premium}
\begin{tabular}{ccccccccc}
\hline
No. &EPS class &$T$&$l_1$&$g_1$&$p_1$&$p_2$&$f_1$&$f_2$ \\
\hline
[3]&Buffer & 1 year & -5\% & 10\% & 0 & 0.5 & 0 & 0.52 \\

[4]&Buffer & 1 year & -5\% & 10\% & 0 & 0.6 & 0 & 0.62 \\

[5]&Buffer & 1 year & -5\% & 10\% & 0 & 0.7 & 0 & 0.72 \\

[6]&Buffer & 1 year & -5\% & 10\% & 0 & 0.8 & 0 & 0.82 \\

[7]&Buffer & 1 year & -5\% & 10\% & 0 & 0.9 & 0 & 0.93 \\

\hline
[7]&Floor & 1 year & -15\% & 10\% & 0.5 & 0 & 0 & 0.56 \\

[8]&Floor & 1 year & -15\% & 10\% & 0.6 & 0 & 0 & 0.67 \\

[9]&Floor & 1 year & -15\% & 10\% & 0.7 & 0 & 0 & 0.78 \\

[10]&Floor & 1 year & -15\% & 10\% & 0.8 & 0 & 0 & 0.89 \\

[11]&Floor & 1 year & -15\% & 10\% & 0.9 & 0 & 0 & 1.0 \\

\hline
\end{tabular}
\label{subtable1:EPS_zeroc}
\end{table}

\newpage

Similarly, if the protection leg parameter $l_1$ for the buffer EPS drops, which means the starting point for a loss protection is lowered, then the fee rate $f_2$ for the EPS provider decreases, with all else equal. For the floor EPS, if $l_1$ decreases, then the protection interval for losses becomes wider. We conclude that the fee rate $f_2$ for an EPS with null fair premium is increasing when the protection leg level $l_1$ decreases (this observation is illustrated in Table \ref{subtable2:EPS_zeroc}).
By combining results in Tables \ref{subtable1:EPS_zeroc} and \ref{subtable2:EPS_zeroc}, we can draw the following conclusion: if the protection strength for losses raises, then the fee rate for all EPSs should increase when the fair premium is postulated
to vanish.

\begin{table}[h!]
\centering
\caption{Buffer and floor EPSs with null fair premium}
\begin{tabular}{ccccccccc}
\hline
No. &EPS class &$T$&$l_1$&$g_1$&$p_1$&$p_2$&$f_1$&$f_2$ \\
\hline
[6]&Buffer & 1 year & -5\% & 10\% & 0 & 0.8 & 0 & 0.82 \\

[10]&Buffer & 1 year & -10\% & 10\% & 0 & 0.8 & 0 & 0.53 \\

\hline
[12]&Buffer & 2 year & -5\% & 10\% & 0 & 0.8 & 0 & 0.69 \\

[13]&Buffer & 2 year & -10\% & 10\% & 0 & 0.8 & 0 & 0.50 \\

\hline
[4]&Floor & 1 year & -5\% & 10\% & 0.8 & 0 & 0 & 0.38 \\

[6]&Floor & 1 year & -10\% & 10\% & 0.8 & 0 & 0 & 0.67\\

[10]&Floor & 1 year & -15\% & 10\% & 0.8 & 0 & 0 & 0.89 \\

\hline
[13]&Floor & 2 year & -5\% & 10\% & 0.8 & 0 & 0 & 0.22 \\

[14]&Floor & 2 year & -10\% & 10\% & 0.8 & 0 & 0 & 0.40 \\

[15]&Floor & 2 year & -15\% & 10\% & 0.8 & 0 & 0 & 0.55 \\

\hline
\end{tabular}
\label{subtable2:EPS_zeroc}
\end{table}

Lastly, we briefly discuss the influence of different values of the fee leg parameter $g_1$. If $g_1$ increases, the starting point of rate of return for which the EPS provider can participate in realised gains is higher. Hence if $g_1$ increases, then the fee rate $f_2$ should increase as well for both the buffer and floor EPS in order to make the fair premium vanish. Rows 2, 6, 11, 12 for the buffer EPS in the top half and rows 2, 4, 12, 13 for the floor EPS in the bottom half of Table \ref{subtable3:EPS_zeroc} are consistent with that finding.

\begin{table}[h!]
\centering
\caption{Buffer and floor EPSs with null fair premium}
\begin{tabular}{ccccccccc}
\hline
No. &EPS class &$T$&$l_1$&$g_1$&$p_1$&$p_2$&$f_1$&$f_2$ \\
\hline
[2]&Buffer & 1 year & -5\% & 5\% & 0 & 0.8 & 0 & 0.61 \\

[6]&Buffer & 1 year & -5\% & 10\% & 0 & 0.8 & 0 & 0.82 \\

\hline
[11]&Buffer & 2 year & -5\% & 5\% & 0 & 0.8 & 0 & 0.56 \\

[12]&Buffer & 2 year & -5\% & 10\% & 0 & 0.8 & 0 & 0.69 \\

\hline
[2]&Floor & 1 year & -5\% & 5\% & 0.8 & 0 & 0 & 0.28 \\

[4]&Floor & 1 year & -5\% & 10\% & 0.8 & 0 & 0 & 0.38 \\

\hline
[12]&Floor & 2 year & -5\% & 5\% & 0.8 & 0 & 0 & 0.18 \\

[13]&Floor & 2 year & -5\% & 10\% & 0.8 & 0 & 0 & 0.22 \\

\hline
\end{tabular}
\label{subtable3:EPS_zeroc}
\end{table}

\subsection{Forward Performance of an EPS}  \label{sec4.3}

We now take the view of the holder and we examine the theoretical forward performance of an EPS during market downturns.
Consistently with the primary purpose of an EPS, its buyer is assumed to be a holder of a superannuation account and thus we assume that she will receive the return $R_T$ on a reference portfolio at maturity date $T$. Recall that the buyer an EPS is required to deliver to its provider the adjusted return $\pp(R_T)$ on the reference portfolio where $\pp(R_T)$ is nonnegative (resp. nonpositive) if the realised return $R_T$ is nonnegative (resp. nonpositive). Therefore, in a study of performance of an EPS we will focus on the realised \textit{net return} for the holder of an superannuation account supplemented by an EPS. Formally, the net return is equal to the difference $R_T - \pp(R_T)$ where $R_T$ is the original return and $\pp(R_T)$ is the adjusted return. In the following example, we consider an investor holding a fair one-year index EPS for five consecutive years, meaning that he/she immediately enters into a new one-year contract upon termination of the previously held one-year index EPS. The simulated performance of the index in the Black-Scholes model and the net returns for the holder of an EPS are presented in Figure~\ref{sim}.

\newpage

We acknowledge that the Black-Scholes model does not offer the most reliable description of market crises but we contend that in a jump-diffusion model the performance of an EPS would be even more visible. To explicitly illustrate the latter conjecture by means of an example, we artificially introduce two negative jumps at deterministic dates in weeks 80 and 250. As before, we take the interest rate $r=1.5\%$, the dividend yield $\kappa=0$, the volatility parameter $\sigma=20\%$, and the initial value $S_0=100$. We perform 200 simulations of sample paths of asset returns over 5 years with the dynamics of the asset price given by \eqref{BSM}. The original and net returns are plotted with the $x$-axis representing the number of trading weeks. Three typical sample paths of returns are shown in black.

In addition, to account for the possibility of market crises, we introduce two down-ticks in weeks 80 and 250, which are indicated by red dashed lines. We then apply five back-to-back one-year fair buffer EPSs with parameters $p_1=0.8,\,l_1=-10\%,\,f_2=0.53$ and $g_1=10\%$. The distributions of the asset's original returns (green line) and net returns for a buffer EPS (blue line) are plotted year by year. Furthermore, we mark the extreme and mean values in the distributions of original and net returns. To enhance clarity, we only show the statistics for the first and fifth years.

\begin{figure} [h!]
    \centering
    \includegraphics[width=16cm, height=8.5cm]{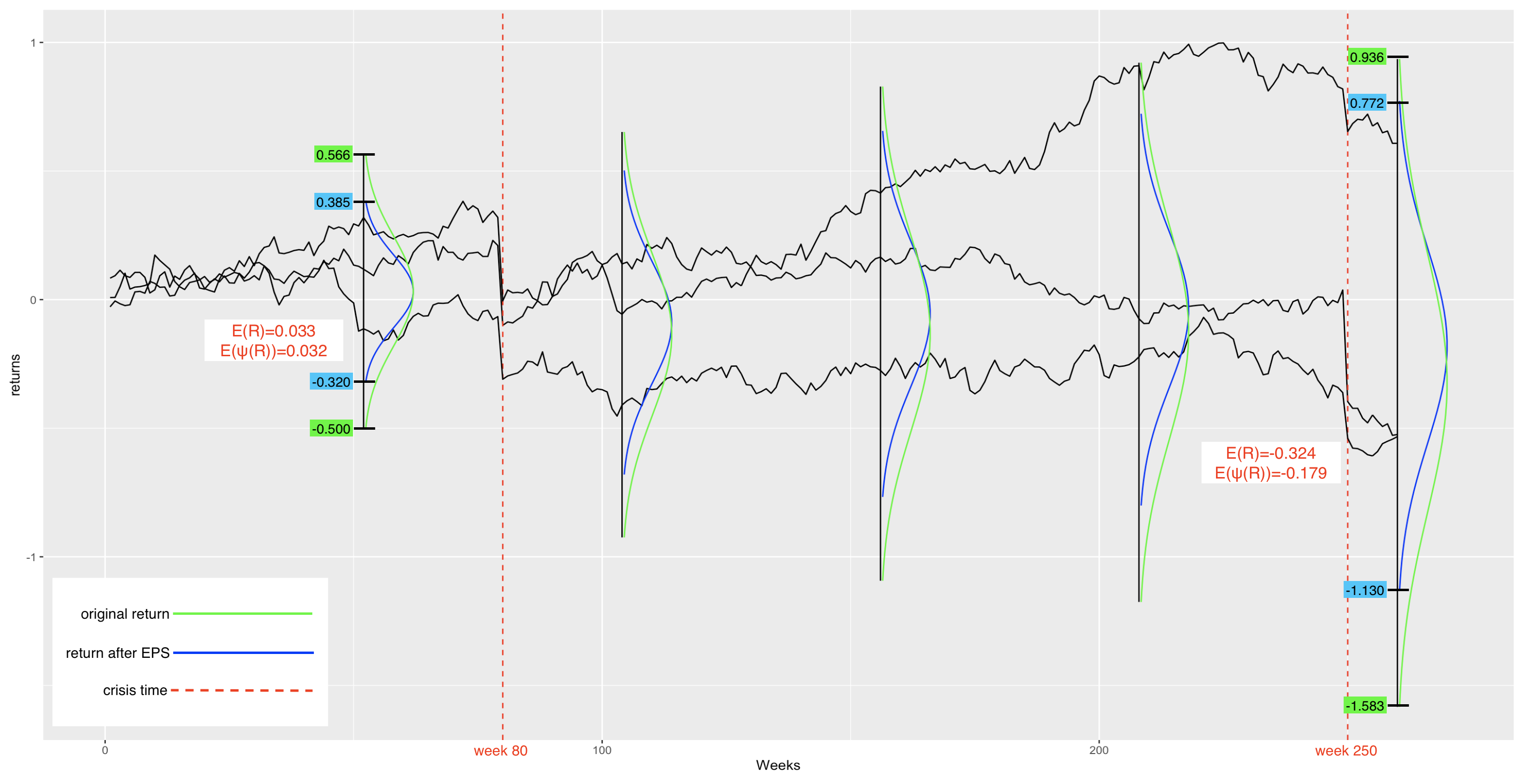}
    \caption{Simulation of original and net returns with a buffer EPS}
    \label{sim}
\end{figure}

In Figure \ref{sim}, we make the first key observation that the distribution curve of net returns is still reminding a normal curve when the buffer EPS is used. It is notable that the expected net returns are visibly higher while the spread narrows toward the mean. This indicates that the net returns have a smaller standard deviation compared to original returns, resulting in a more concentrated distribution with a higher peak. If no unpredictable events happen in the first year, holders of the EPS are protected from a loss less than $l_1$, and, as a trade-off, she/he needs to share any benefits greater than the threshold $g_1$.

Furthermore, the impact of the EPS on benefits sharing and loss protection is exactly the same in relation to the market factors and the participation rates $p_1$ and $f_2$. This is because the considered EPS is fair to both parties. The extreme positive values of the original and net returns with the buffer EPS at the end of the first year are $56.6\%$ and $38.5\%$, respectively, while the extreme negative values are $-50.0\%$ and $-32.0\%$, respectively. Additionally, for risk-averse investors, the buffer EPS can offer even more protection against losses if an insurance premium is paid to the provider. At the end of the fifth year, after experiencing two unpredictable market crises, the expected value of original returns drops from $0.033\%$ to $-0.324\%$ which represents the worse case scenarios for investors in real life. However, with the help of five one-year buffer EPSs, the mean of net returns was lifted to $-17.4\%$, which is much higher than the mean of original returns at the end of the fifth year ($-32.4\%$). The entire net returns curve is shifted upward towards a more concentrated distribution, which results in the higher mean and significantly elevated negative returns.

On the gains side, the extreme positive value of net returns decreased only from $93.6\%$ (original returns) to $77.2\%$ (net returns), which is a difference of $16.4\%$. In contrast, compared to the extreme negative value of the original return ( $-158.3\%$), the net returns rose to $-113\%$, an increase of $45.3\%$, nearly three times the difference of the extreme positive returns. The occurrence of market downturns accentuates the prominence of the protection leg over the fee leg, leading to asymmetric contraction levels in the two legs. These results underscore the role of EPSs as an investment insurance, demonstrating their enhanced performance in periods of increased market risks, particularly in the presence of market crises.

\section{Backtesting of Performance of Buffer and Floor EPSs} \label{sec5}

In our numerical analysis of some classes of equity protection swaps, we computed the theoretical fair premia for several EPSs with predetermined participation rates, as outlined in Section \ref{sec4.1}. Subsequently, we have examined the level of the fee rate for which the fair premium $\whc=\whH_0$ vanishes.  In the final step, we analyse the impact of an EPS on holder's investment through backtesting using the real-world data for two major stock indices: S\&P~500 and S\&P/ASX~200 in the U.S. and Australia, respectively. It should be acknowledged that we do not consider the currency risk and thus backtesting is done independently for the domestic and foreign investments, denominated in AUD and USD, respectively.
It is important to stress that throughout this section, all of our discussions and pricing results are model-free, in the sense that we do not make any assumptions about the valuation models or the future dynamics of indices.

We argued that, unlike other insurance products for variable annuities, the fair pricing of basic EPS products can be considered to be model-free since their valuation can be done directly in reference to the market data for traded options without using any stochastic model for future dynamics of the underlying index. Specifically, the fair premium for an EPS can be calculated using directly the market prices of the European call and put options with weights given by a static hedging strategy. Needless to say, since the fair premia (or, equivalently, the fair fee rates) are implicitly given by the real market data, they necessarily vary when market conditions change.
In our empirical study of the real-world benefits of an EPS for its holder, we will focus on the following six specifications for buffer and floor EPSs:
\begin{itemize}
    \item \textbf{Buffer 1}: $p_2=0.5,\ f_2=0.63,\ l_1=-5\%,\ g_1=5\%$;
    \item \textbf{Buffer 2}: $p_2=0.7,\ f_2=1.51,\  l_1=-5\%,\ g_1=10\%$;
    \item \textbf{Buffer 3}: $p_2=0.7,\ f_2=1.21,\ l_1=-10\%,\ g_1=10\%$;
    \item \textbf{Floor 1}:  $p_1=0.5,\ f_2=0.52,\ l_1=-10\%,\ g_1=10\%$;
    \item \textbf{Floor 2}:  $p_1=0.7,\ f_2=0.73,\ l_1=-10\%,\ g_1=10\%$;
    \item \textbf{Floor 3}:  $p_1=0.7,\ f_2=0.98,\ l_1=-15\%,\ g_1=10\%$.
\end{itemize}

These six products were postulated to be fair EPSs, meaning that their initial premia at inception should vanish. To this end, we first selected the protection rate ($p_1$ or $p_2$), as well as the protection and fee thresholds ($l_1$ and $g_1$) and then we used the market data for European options to identify a unique level of the fee rate $f_2$, which makes each EPS a fair swap. As was mentioned, the initial fair fee rate $f_2$ varies from day to day but here, for the sake of simplicity, we have only used one day market data for European options on S\&P~500 index to compute the fair fee rate $f_2$. The date (2022.02.02) was chosen as the middle date of the period from 3 May 2021 to 23 December 2022 when the market downturn
has occurred. The closing prices for in-the-money and at-the-money European put and call options on S\&P~500 index with a one-year maturity are reported in Table \ref{data_option}. We give there the quote and expiration dates, the strike price and the spot level of S\&P~500 index, and bid and ask prices for relevant options.  Notice that the values of the moneyness, here defined by convention as the ratio strike/spot are aligned with the value $1+l_1$ for put options and $1+g_1$ for call options (recall that $K^l_1=S_0(1+l_1)$ and $K^g_1=S_0(1+g_1)$).

\begin{table}[h!]
\centering
\caption{Market data for European options on S\&P~500 index}
\begin{tabular}{ccccccccc}
\hline
No. &Quote Date &Expiration &Strike &Type &Bid &Ask &Spot &Moneyness \\
\hline
[1]&2022/02/02 & 2023/02/17 & 3900 & Put & 184.6 & 187.8 & 4576.8 & 85.2\% \\

[2]&2022/02/02 & 2023/02/17 & 4125 & Put & 235.8 & 239.8 & 4576.8 & 90.1\% \\

[3]&2022/02/02 & 2023/02/17 & 4350 & Put & 298.3 & 302.7 & 4576.8 & 95.0\% \\

[4]&2022/02/02 & 2023/02/17 & 4575 & Put & 375.2 & 379.8 & 4576.8 & 100\% \\

[5]&2022/02/02 & 2023/02/17 & 4575 & Call & 366.1 & 370.4 & 4576.8 & 100\% \\

[6]&2022/02/02 & 2023/02/17 & 4800 & Call & 239.0 & 243.4 & 4576.8 & 104.9\% \\

[7]&2022/02/02 & 2023/02/17 & 5025 & Call & 139.0 & 142.8 & 4576.8 & 109.8\% \\

\hline
\end{tabular}
\label{data_option}
\end{table}

We use the same six EPS products in empirical studies for both S\&P~500 index and S\&P/ASX~200 index, although historical market data for the European option on S\&P/ASX~200 index were not available. Hence it was assumed that the market prices of European options for both indices will be close. In practice, our approach can be easily extended to the fully dynamic fair pricing of EPSs but we contend that the conclusions from our study would not change substantially.


\subsection{Backtesting During Market Downturn}  \label{sec5.1}

To assess the benefits from holding an EPS, we initially focus on scenarios in which the market experiences a downturn, which means that the underlying assets have negative returns. We assume that a cohort of investors engage on any given trading day in one-year EPS contracts of a given type and the EPSs commenced on any given day have the same combined notional principal, the size of which is in fact irrelevant for our conclusions.

Regarding the market data, we use the daily closing prices of S\&P~500 and S\&P/ASX~200 indices from 2 January 2020 to 23 December 2022 to determine the realised trailing returns on the index serving as a reference portfolio for an EPS under study. In order to analyse the impact of predominantly negative returns, we collected 164 one-year trailing returns of the S\&P~500 index and 165 one-year trailing returns of the S\&P/ASX~200 index from 3 May 2021 to 23 December 2022. The histograms in Figure \ref{hist_sp_asx_neg} present empirical frequencies of realised one-year trailing returns on S\&P~500 and S\&P/ASX~200 indices during the considered period.

\begin{figure}
  \centering
  \begin{tabular}[b]{c}
    \includegraphics[width=.43\linewidth]{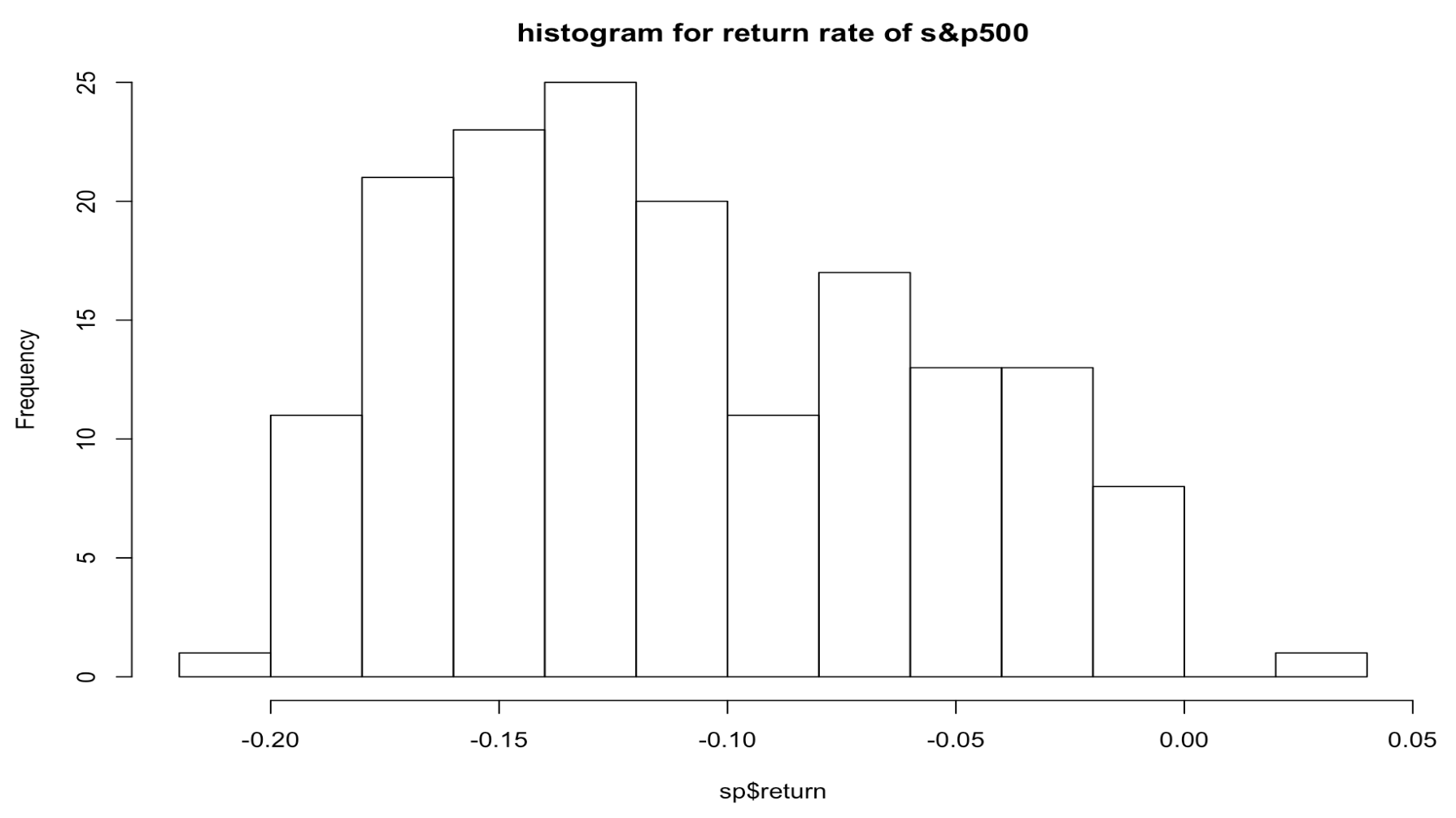} \\
    \small (a) S\&P~500 one year trailing return
  \end{tabular} \qquad
  \begin{tabular}[b]{c}
    \includegraphics[width=.43\linewidth]{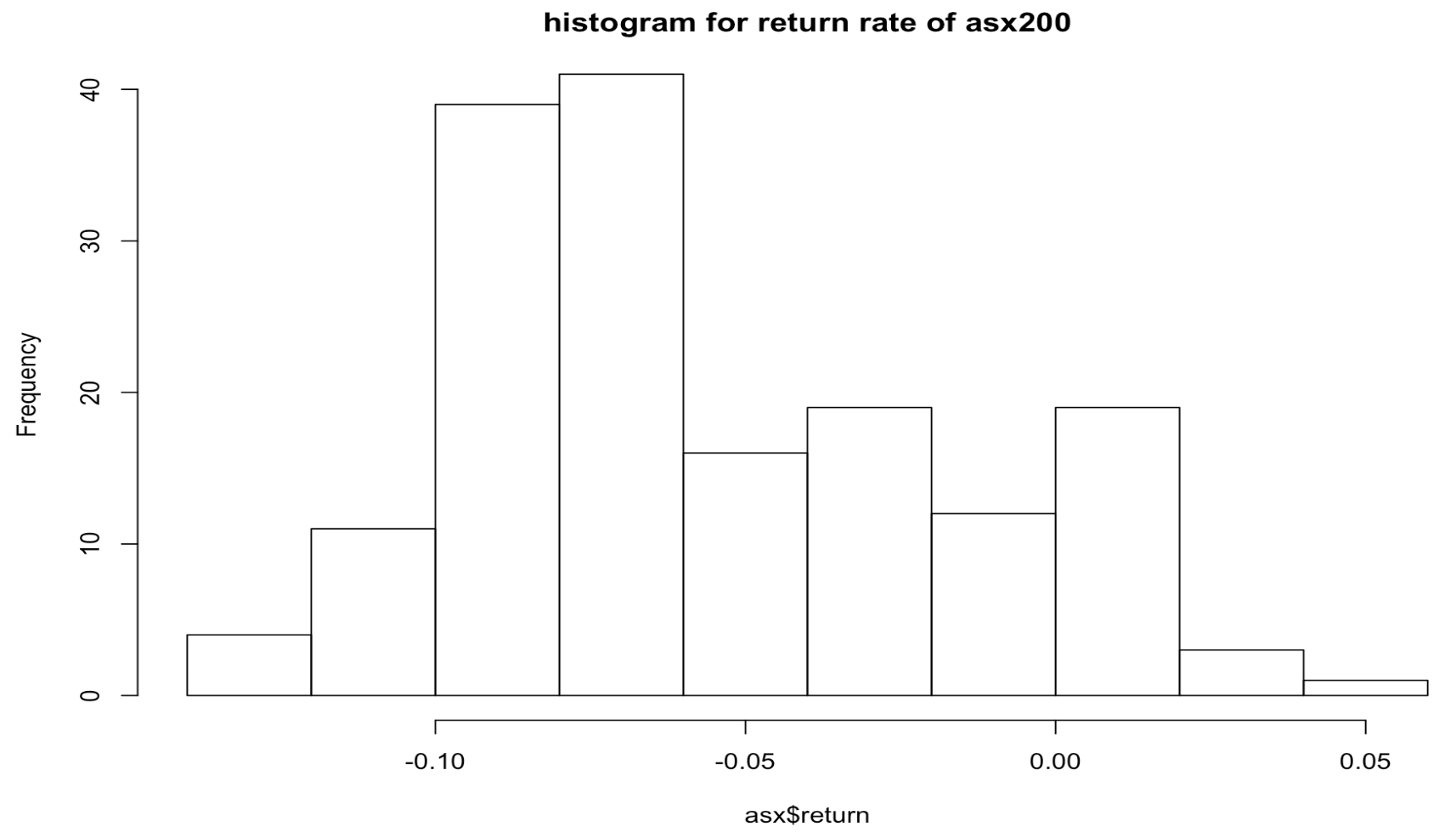} \\
    \small (b) S\&P/ASX~200 one year trailing return
  \end{tabular}
\caption{Histograms of trailing returns on S\&P~500 and S\&P/ASX~200 indices}\label{hist_sp_asx_neg}
\end{figure}

\begin{figure} [h!]
    \centering
    \includegraphics[width=14cm, height=8.5cm]{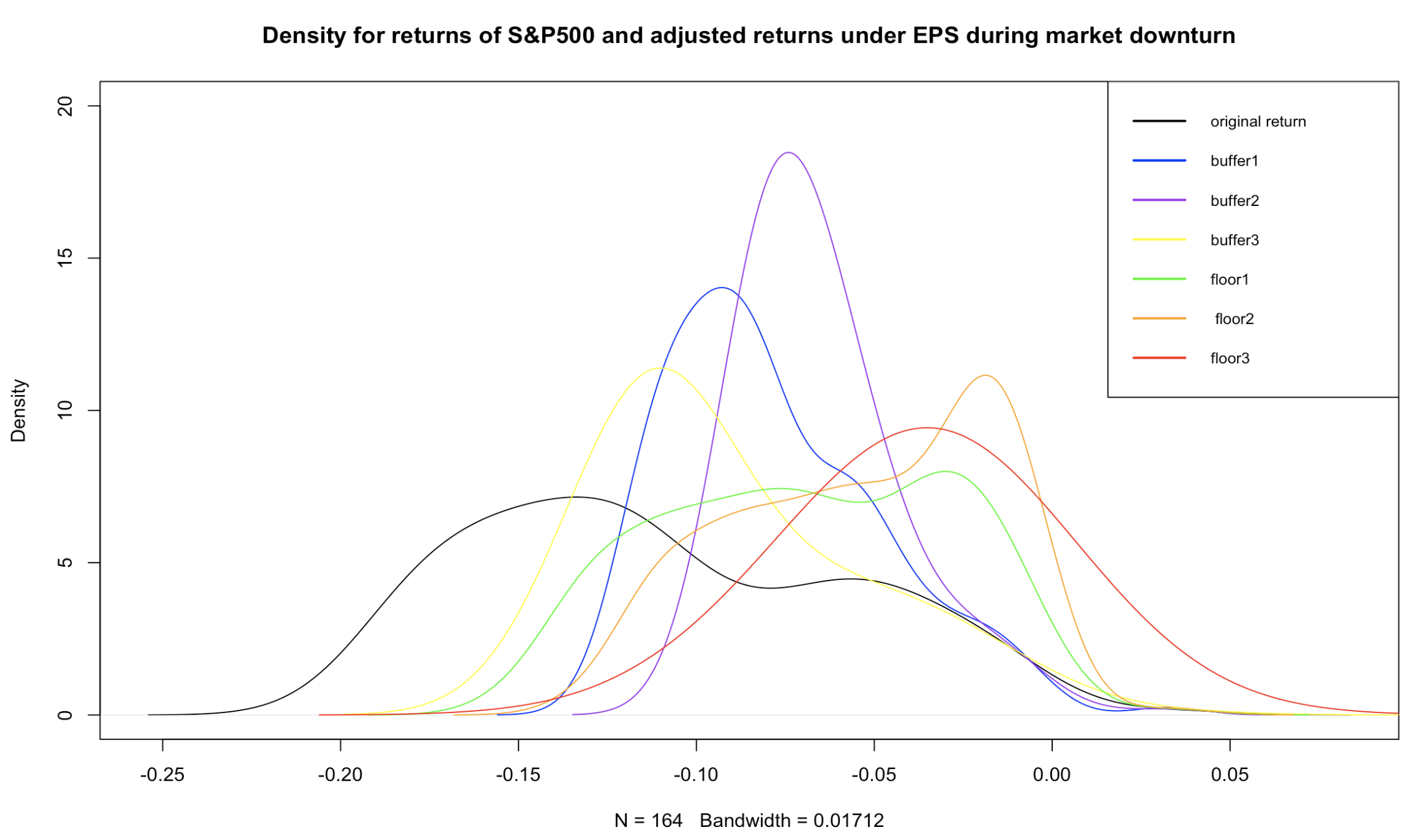}
    \caption{Empirical densities of original and net trailing returns on S\&P~500 under market downturn}
    \label{back_sp_neg}
\end{figure}

\begin{table}[h!]
\centering
\caption{Empirical quantiles of original and net trailing returns on S\&P~500 under market downturn}
\begin{tabular}{cccccccccc}
\hline
Case & Min & $5\%$ & $10\%$ & $25\%$ & $50\%$  & $75\%$ & $90\%$ & Max \\
\hline
Original & -0.2027 & -0.1822 & -0.1777 & -0.1543 & -0.1180 & -0.0649 & -0.0365 & 0.0325\\
\hline
Buffer1 & -0.1264 & -0.1161 & -0.1139 & -0.1021 & -0.0840 & -0.0575 & -0.0365 & 0.0325  \\
\hline
Buffer2 & -0.0958 & -0.0897 & -0.0883 & -0.0813 & -0.0704 & -0.0545 & -0.0365 & 0.0325 \\
\hline
Buffer3 & -0.1308 & -0.1247 & -0.1233 & -0.1163 & -0.1054 & -0.0649 & -0.0365 & 0.0325  \\
\hline
Floor1 & -0.1527 & -0.1322 & -0.1277 & -0.1043 & -0.0680 & -0.0325 & -0.0183 & 0.0325  \\
\hline
Floor2 & -0.1327 & -0.1122 & -0.1077 & -0.0843 & -0.0480 & -0.0195 & -0.0110 & 0.0325 \\
\hline
Floor3 & -0.0977 & -0.0772 & -0.0727 & -0.0493 & -0.0354 & -0.0195 & -0.0110 & 0.0325 \\
\hline
\end{tabular}
\label{quan_sp_neg}
\end{table}

From Figure \ref{hist_sp_asx_neg} it transpires that the realised one-year trailing returns in the considered period are characterised by a negative mean value. By combining either a buffer or a floor EPS with the original returns, we obtain the sample of net trailing returns on S\&P~500 index. To compare the original and net trailing returns, we display in Figure \ref{back_sp_neg} their respective empirical densities and in Table \ref{quan_sp_neg} we give the values of empirical quantiles.
Figure \ref{back_sp_neg} and Table \ref{quan_sp_neg} reveal that, on the negative side, the net trailing returns dominate the original trailing returns, have a less negative empirical mean, and exhibit an empirical distribution not far away from the normal distribution. It is worth noting that the buffer EPS No. 2 in Table \ref{quan_sp_neg} offers the best protection among three buffer contracts, which is due to the fact that it has a higher protection rate ($p_2=0.7$) than the buffer EPS No. 1 ($p_2=0.5$), and a higher protection threshold ($l_1=-5\%$) than the buffer EPS No. 3 ($l_1=-10\%$). Similarly, the floor EPS No. 3 offers the highest protection among three floor contracts, has a higher protection rate ($p_1=0.7$) than the floor No. 1 ($p_1=0.5$), and a lower floor ($l_1=-15\%$) than the floor EPS No. 2 ($l_1=-10\%$).
In Figure \ref{back_sp_neg}, we observe that floor EPSs shift negative original trailing returns to higher levels than buffer EPSs when considering the market during a downturn period. This feature was in fact expected, since a buffer EPS only protects against sizeable negative returns, while a floor EPS provides protection against any negative return.

We have also examined the impact of EPSs on the S\&P/ASX~200 index using 165 one-year realised trailing returns on S\&P/ASX~200. Figure \ref{back_asx_neg} displays the empirical densities of one-year trailing returns on the S\&P/ASX~200 index and net returns for various EPS products. Additionally, in Table \ref{quan_asx_neg} we report the empirical quantiles of S\&P/ASX~200 index returns for the period under study.

\begin{figure} [h!]
    \centering
    \includegraphics[width=14cm, height=8.5cm]{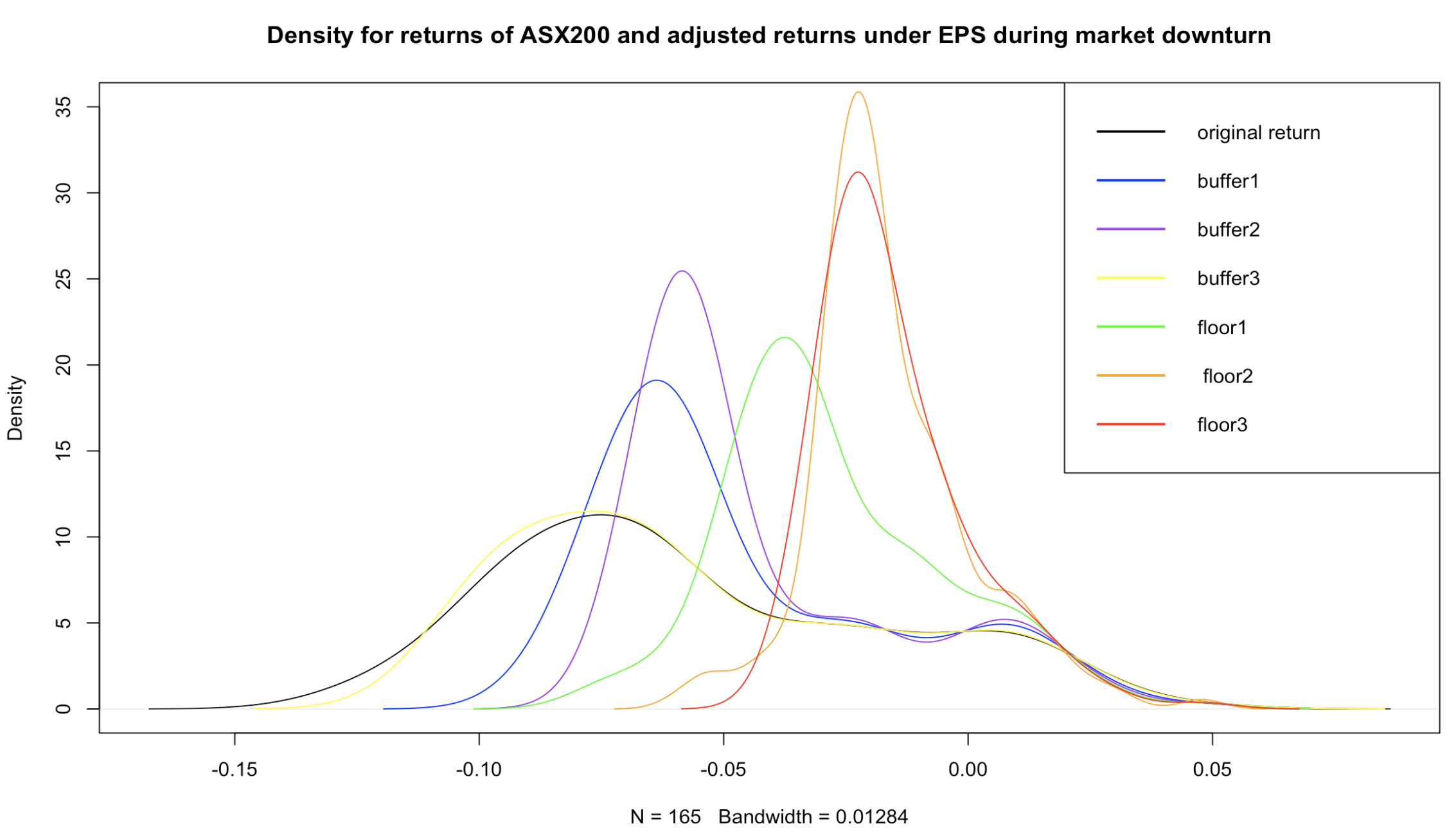}
    \caption{Empirical densities of original and net trailing returns on S\&P/ASX~200 under market downturn}
    \label{back_asx_neg}
\end{figure}

\begin{table}[h!]
\centering
\caption{Empirical quantiles of original and net trailing returns on S\&P/ASX~200 under market downturn}
\begin{tabular}{cccccccccc}
\hline
Case & Min & $5\%$ & $10\%$ & $25\%$ & $50\%$  & $75\%$ & $90\%$ & Max \\
\hline
Original & -0.1290 & -0.1092 & -0.0988 & -0.0858 & -0.0654 & -0.0248 & 0.0076 & 0.0478 \\
\hline
Buffer1 & -0.0895 & -0.0796 & -0.0744 & -0.0679 & -0.0577 & -0.0248 & 0.0076 & 0.0478 \\
\hline
Buffer2 & -0.0737 & -0.0678 & -0.0646 & -0.0607 & -0.0546 & -0.0248 & 0.0076 & 0.0478  \\
\hline
Buffer3 & -0.1087 & -0.1028 & -0.0988 & -0.0858 & -0.0654 & -0.0248 & 0.0076 & 0.0478  \\
\hline
Floor1 & -0.0790 & -0.0592 & -0.0494 & -0.0429 & -0.0327 & -0.0124 & 0.0076 & 0.0478  \\
\hline
Floor2 & -0.0590 & -0.0392 & -0.0297 & -0.0257 & -0.0196 & -0.0074 & 0.0076 & 0.0478  \\
\hline
Floor3 & -0.0387 & -0.0328 & -0.0296 & -0.0257 & -0.0196 & -0.0074 & 0.0076 & 0.0478  \\

\hline
\end{tabular}
\label{quan_asx_neg}
\end{table}

It is no surprising to see in Figure \ref{back_asx_neg} and Table \ref{quan_asx_neg} that the relationship between different EPS contracts is similar to that observed in Figure \ref{back_sp_neg} and Table \ref{quan_sp_neg}. We conclude that, regardless of the data used, the benefits of EPS contracts for portfolio's protection are consistent.

It can also be observed that a floor EPS has a stronger impact on S\&P/ASX~200 index returns yielding relatively more concentrated net returns, while a buffer EPS typically shows a more pronounced influence on S\&P~500 index returns. This phenomenon is due to the fact that a floor EPS protects against any losses of a limited size, whereas a buffer EPS can only protect against relatively large losses but is unlimited. A buffer EPS would be preferred by a holder characterised by a high level of risk aversion to large losses, whereas a floor EPS would be more appropriate for a holder expecting only relatively low losses during the holding period of an EPS. Finally, since original trailing returns on S\&P/ASX~200 index were broadly less negative compared to those on S\&P~500 index, a floor EPS would be more suitable when the S\&P/ASX~200 index is chosen as the reference portfolio.

\subsection{Backtesting During Market Downturn and Upturn} \label{sec5.2}

Next, we study the impact of EPS contracts on a general scenario where the real-world market has experienced both a slump and an upsurge. To this end, we analyse the daily closing prices of S\&P~500 and S\&P/ASX~200 between 2 January 2020 and 23 December 2022, resulting in 499 one-year trailing returns for the S\&P~500 index and 502 one-year trailing returns for the S\&P/ASX~200 index. Figure \ref{hist_sp_asx} shows the empirical histograms of one-year trailing returns on S\&P~500 and S\&P/ASX~200 indices. We observe that the average trailing return on the S\&P~500 index is higher than that on the S\&P/ASX~200 index and, furthermore, the realised trailing returns on S\&P~500 index are characterised by a higher volatility.

\begin{figure}
  \centering
  \begin{tabular}[b]{c}
    \includegraphics[width=.43\linewidth]{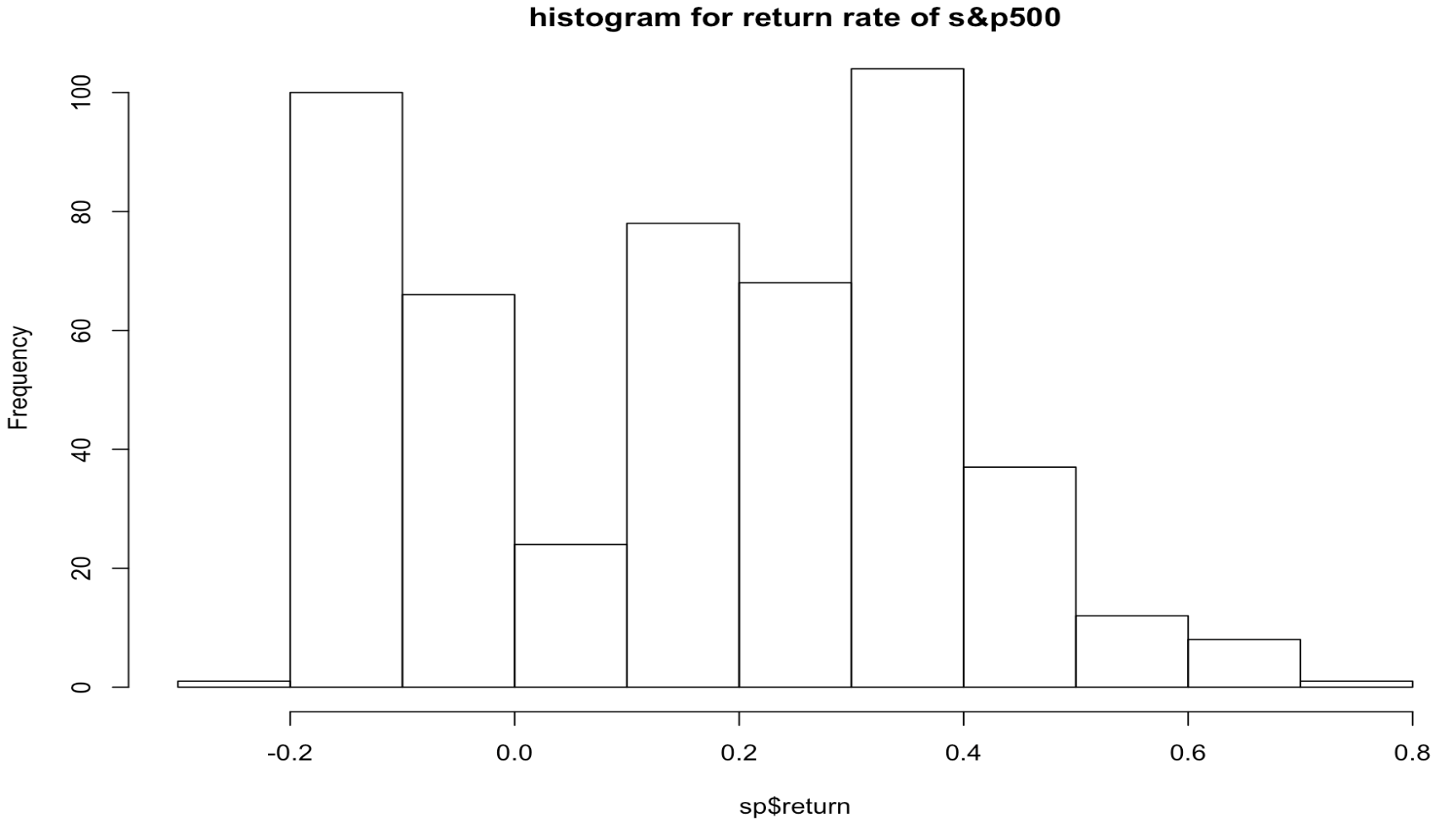} \\
    \small (a) S\&P~500 one year trailing return
  \end{tabular} \qquad
  \begin{tabular}[b]{c}
    \includegraphics[width=.43\linewidth]{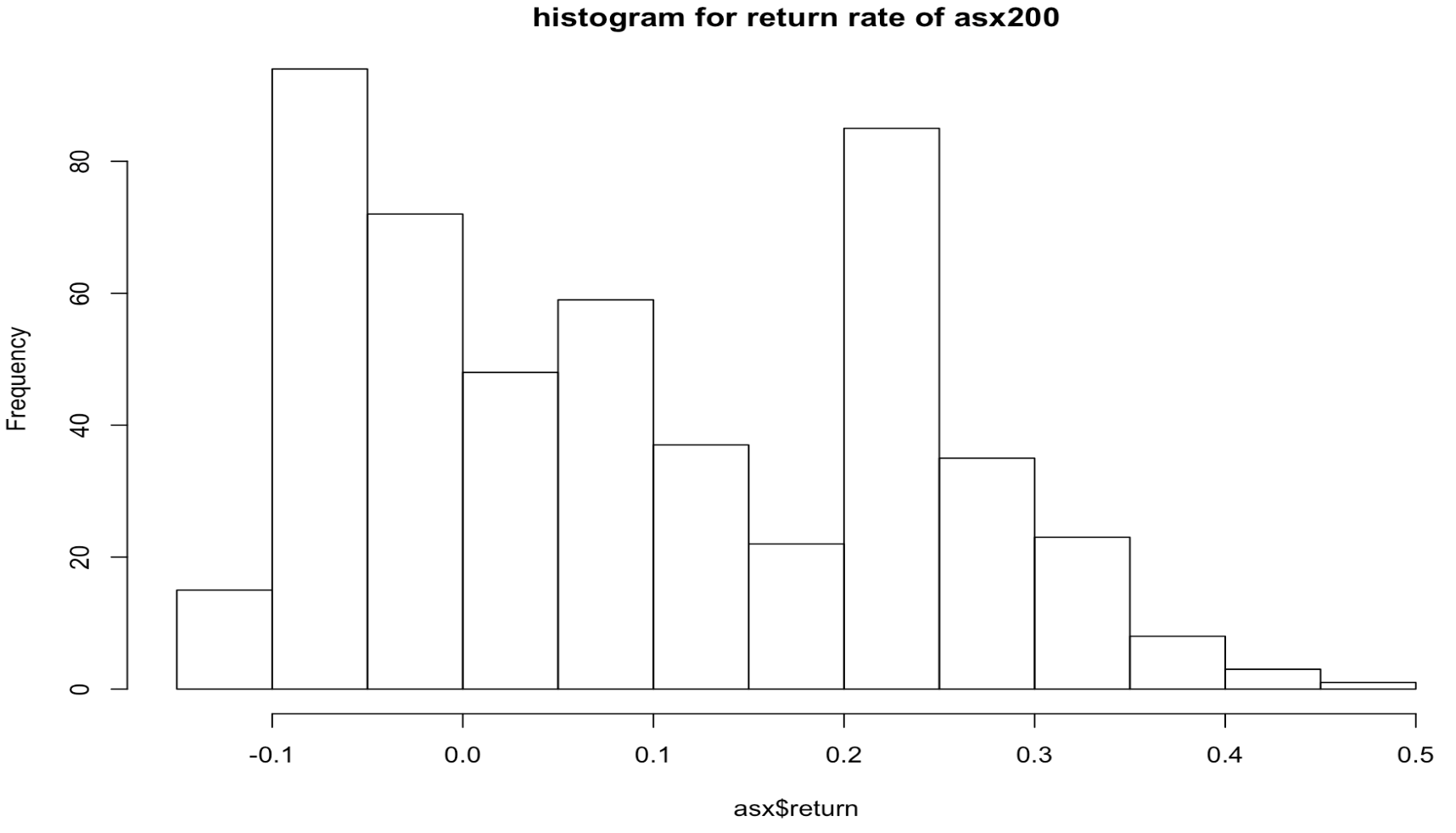} \\
    \small (b) S\&P/ASX~200 one year trailing return
  \end{tabular}
\caption{Histograms of trailing returns for S\&P~500 and S\&P/ASX~200 indices}\label{hist_sp_asx}
\end{figure}

To examine the effect of an EPS on both positive and negative returns, we continue to examine six EPSs (three buffers and three floors) introduced at the beginning of this section to compute the net returns for the holder of an EPS.  Figure \ref{back_sp_whole} illustrates the empirical densities for original and net returns on the S\&P~500 index and Table \ref{quan_sp_whole} gives the respective empirical quantiles.

\begin{figure} [h!]
    \centering
    \includegraphics[width=14cm, height=8.5cm]{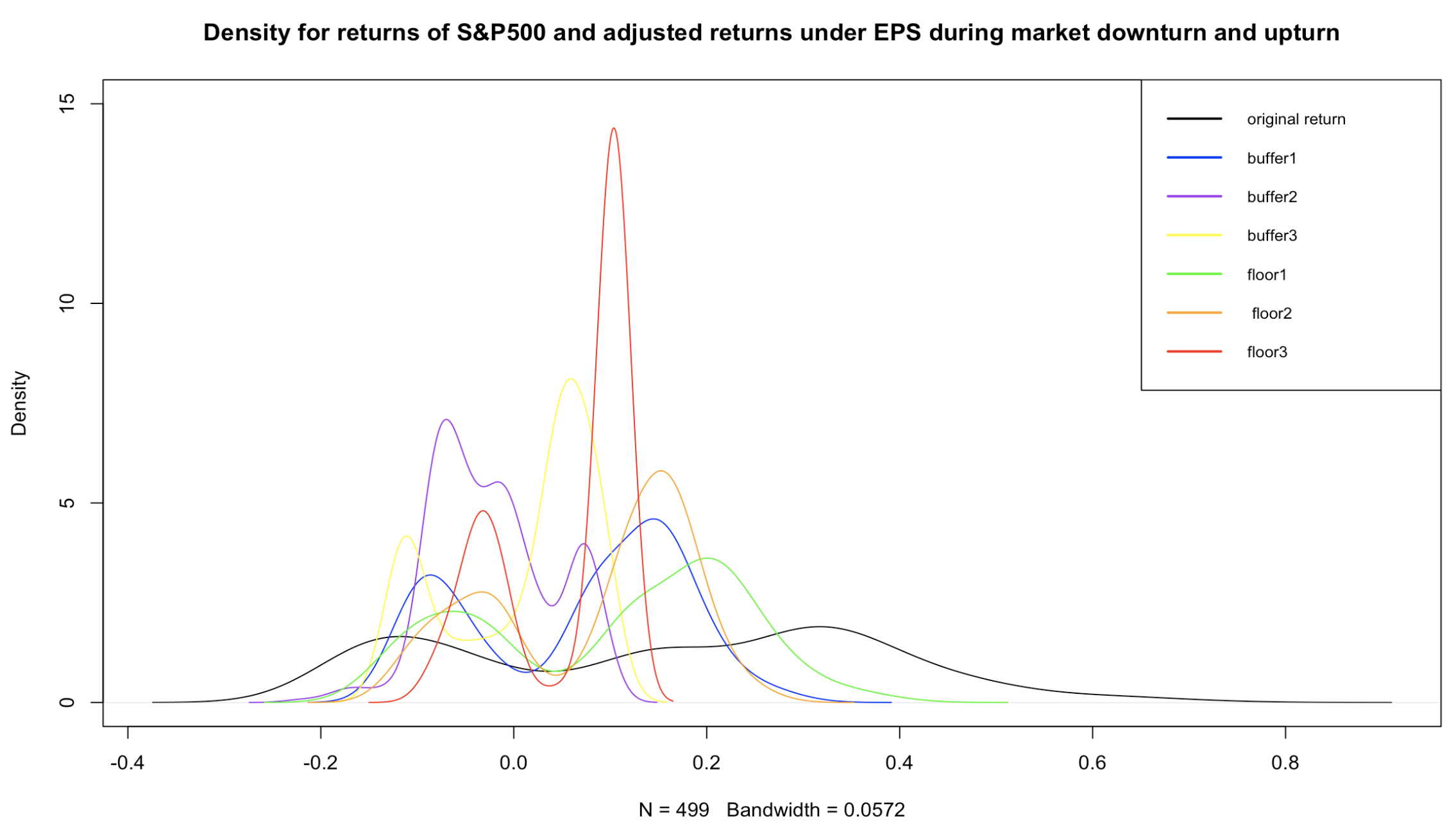}
    \caption{Empirical densities of trailing returns on S\&P~500 under market downturn and upturn}
    \label{back_sp_whole}
\end{figure}

\begin{table}[h!]
\centering
\caption{Empirical quantiles of original and net trailing returns on S\&P~500 under market downturn and upturn}
\begin{tabular}{cccccccccc}
\hline
Case & Min & $5\%$ & $10\%$ & $25\%$ & $50\%$  & $75\%$ & $90\%$ & Max\\
\hline
Original & -0.2027 & -0.1654 & -0.1444 & -0.0617 &  0.1641 & 0.3236 &  0.4172 & 0.7382  \\
\hline
Buffer1 & -0.1264 & -0.1077 & -0.0972 & -0.0559 & 0.0922 & 0.1512 &  0.1859 & 0.3046  \\
\hline
Buffer2 & -0.2255 & -0.0969 & -0.0880 & -0.0704 & -0.0269 &  0.0219 & 0.0726 & 0.0997 \\
\hline
Buffer3 & -0.1308 & -0.1196 & -0.1133 & -0.0617 & 0.0427 &0.0655  &0.0882 & 0.0997 \\
\hline
Floor1 & -0.1527 & -0.1154 & -0.0944 & -0.0309 & 0.1308 & 0.2073 & 0.2523 & 0.4063 \\
\hline
Floor2 & -0.1327 & -0.0954 & -0.0744 & -0.0185 & 0.1173 & 0.1604 & 0.1857 & 0.2723 \\
\hline
Floor3 & -0.0977 & -0.0604 & -0.0433 & -0.0185 & 0.1013 &0.1045 &  0.1063 & 0.1128 \\
\hline
\end{tabular}
\label{quan_sp_whole}
\end{table}

In addition, Figure \ref{back_asx_whole} presents the empirical densities for the original and net trailing returns on the S\&P/ASX~200 index, accompanied by their empirical quantiles given in Table \ref{quan_asx_whole}. Observe that the empirical densities of returns on the S\&P/ASX~200 index exhibit similar main characteristics as those obtained for the S\&P~500 index. While we previously examined the impact of EPSs on negative returns in the benefit testing for investors, we now focus on their influence on positive returns.

\begin{figure} [h!]
    \centering
    \includegraphics[width=14cm, height=8.5cm]{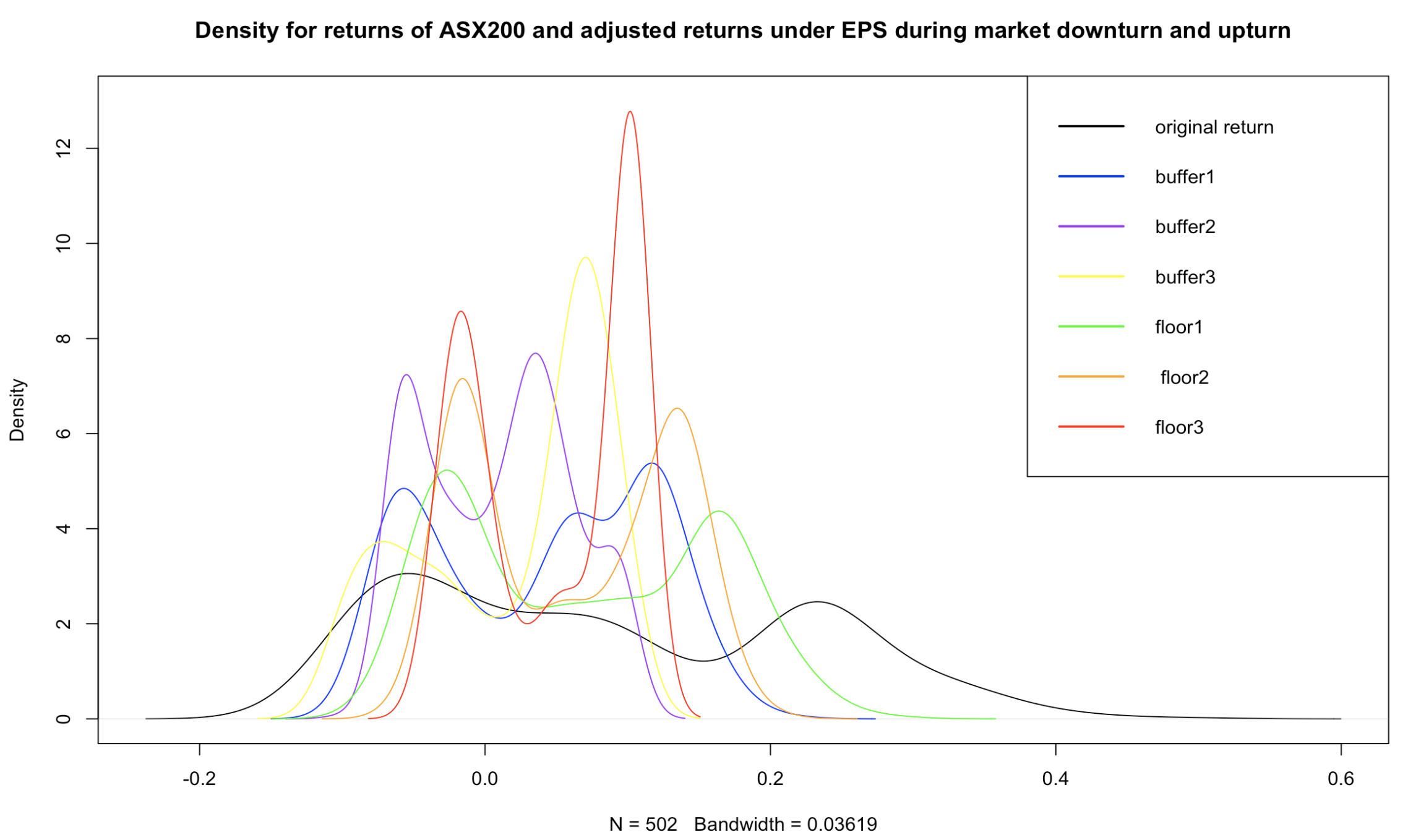}
    \caption{Empirical densities of trailing returns on S\&P/ASX~200 under market downturn and upturn}
    \label{back_asx_whole}
\end{figure}

\begin{table}[h!]
\centering
\caption{Empirical quantiles of original and net trailing returns on S\&P/ASX~200 during market downturn and upturn}
\begin{tabular}{cccccccccc}
\hline
Case & Min & $5\%$ & $10\%$ & $25\%$ & $50\%$  & $75\%$ & $90\%$ & Max \\
\hline
Original & -0.1290 & -0.0939 & -0.0804 & -0.0386 & 0.0630 & 0.2239 &  0.2690 & 0.4912  \\
\hline
Buffer1 & -0.0895 & -0.0719 & -0.0652 & -0.0386 & 0.0548 & 0.1144 & 0.1310 & 0.2132 \\
\hline
Buffer2 & -0.0995 & -0.0635 & -0.0593 & -0.0418 & 0.0165 &  0.0475& 0.0837 & 0.0999 \\
\hline
Buffer3 & -0.1087 & -0.0939 & -0.0804 & -0.0386 & 0.0509 & 0.0730 &  0.0915 & 0.0999  \\
\hline
Floor1 & -0.0790 & -0.0469 & -0.0402 & -0.0193 & 0.0630 & 0.1595 &  0.1811 & 0.2877 \\
\hline
Floor2 & -0.0590 & -0.0282 & -0.0241 & -0.0116 & 0.0630 &0.1335 & 0.1456 & 0.2056 \\
\hline
Floor3 & -0.0387 & -0.0282 & -0.0241 & -0.0116 & 0.0630 &0.1025 & 0.1034 & 0.1078 \\
\hline
\end{tabular}
\label{quan_asx_whole}
\end{table}

When comparing the empirical densities of original and net trailing returns on S\&P~500 and S\&P/ASX~200 indices, we can observe that the net returns are more tightly distributed and exhibit a smaller variance. Furthermore, the frequencies of positive and negative values for both the original and net trailing returns are similar, with positive (resp., negative) net returns being systematically lower (resp., higher) for the net returns. It is common in the economic literature to assume that investors are risk-averse and pay more attention to losses, rather than gains. Although the net trailing returns have compressed positive parts, there are ample reasons to contend that EPS products would prove advantageous for holders of superannuation accounts since their protection leg is an effective tool to mitigate portfolio's losses.

To summarise, the empirical densities of net returns for the holder of an EPS presented in Figures \ref{back_sp_whole} and \ref{back_asx_whole} exhibit sharper peaks than the original returns, which indicates that an EPS offers effective protection against substantial losses at a cost of an insurance fee in the form of a minor reduction of large gains. As expected, a holder of an EPS receives a higher return in the case of portfolio's loss combined with a lower return when gain occurs, which helps to mitigate the uncertainty in portfolio's future values.

\section{Concluding Remarks} \label{sec6}

This work examines the performance of a new class of financial derivatives, tentatively called \textit{equity protection swaps} (EPS), whose aim is to provide tailored-made protection for individual superannuation accounts. A generic EPS is a special kind of equity swap that combines the features of a total return swap with a customised contract that allows both the holder and provider to select levels of the protection and fee legs, as well as the corresponding participation rates.

An EPS is purposely designed to have a relatively simple structure in comparison to other existing guarantee products for variable annuities, while it still retains the important attribute of providing at least a partial guarantee for the value of a reference portfolio at no initial cost. Therefore, we contend that they are likely to have a huge potential for holders of superannuation accounts.
In essence, the provider of an EPS receives potential gains when the value of the reference portfolio increases between the contract's inception date and its maturity and offers partial protection to the holder when the value of the reference portfolio drops during that period. Since an EPS can be easily customized, it should be attractive to portfolio holders as a tool to mitigate losses at a desired level of protection based on their individual preferences.  

Additionally, an EPS offers a unique structure that allows for the fair premium to be easily calculated for varying maturities, protection thresholds and rates, as well as fee thresholds and rates. The existence of an explicit formula for the fair premium is beneficial for both the EPS provider and buyer. According to behavioral finance, investors tend be unwilling to pay a large amount upfront in exchange for insurance of a potential loss in the future. To address this issue, the provider has also the option to adjust the fee rate so that the fair initial premium is null. In that case the only cost for investors is a predetermined share of potential gains at maturity, which should make an EPS more attractive to them. The flexibility and simplicity of an EPS make it an ideal product for the pension market, which is characterised by unsophisticated investors
whose livelihoods depend on income stream drawn from their superannuation accounts.

From the provider's perspective, we argued that the use of European call and put options is an effective method for the static (and hence model-free) hedging strategy for an EPS. By using call options to hedge potential gains in the fee leg and put options to hedge potential losses in the protection leg, providers can create a well-structured static hedging portfolio. Two specific structures of an EPS, the buffer EPS and the floor EPS, were used throughout to explicitly illustrate the hedging portfolios but the method works for other cases as well. Of these two structures, the floor EPS was found to have a superior hedging structure, due to its cap on potential provider's losses in case of counterparty credit risk of issuers of put options. While the present work focuses on these two specific structures, other specification, such as
buffer-floor/buffer-cap EPS of Example \ref{buffer_ex}, will be explored in future research.

Another important advantage of an EPS is the fact that it can be offered by a third-party provider, as this completely eliminates the pressure on managing the underlying reference portfolio. However, in such a situation, it would be important to take into account the provider's default risk since it may have significant consequences for the buyers. In most cases, the holder of an EPS may not default even if there is a gain because their accounts are managed by superannuation funds. In contrast, the third-party providers may also default in some extreme events, which means they could not provide enough protection for the holder's losses. As such, it is recommended that superannuation funds offer EPS contracts to their members, as this reduces the default probability given their huge cash flows from variable annuities. To reduce the risk exposure in the event of a third-party default, a collateralised variant of an EPS should also be examined in future studies.

We acknowledge that we presented here only an introduction to these novel investment insurance products for Australian superannuation system. We have limited ourselves to a study of some simple structures of the standard EPS by focusing on real-life advantages of various classes of EPSs.
As explained in Subsection \ref{sec2.4}, the most basic product is an index EPS but, in practice, other classes of EPSs are likely to be of interest and it should be noted that not all EPSs can be fairly priced with no reference to some stochastic models for indices or exchange rates. One could also examine a more complex structure of an EPS, tentatively called a {\it two-phased EPS}, which offers a more comprehensive risk management solution for a market crash at predetermined postponed settlement date $T_p > T$, rather than the original maturity date $T$.  Due to the expected market recovery between the maturity date $T$ and the final settlement date $T_p$, the holder of a two-phased EPS would obtain partial protection against a market crash. The postponed settlement is important since it would give both the provider and holder more opportunities to mitigate their respective loss exposures during a market downturn. Obviously, the concept of two-phased EPS adds a new dimension to the analysis of an EPS, making it more challenging, but also more useful in addressing the financial risks faced by holders of superannuation accounts.

Another crucial area for future research is the consideration of jumps in the reference portfolio. The limitations of the Black-Scholes model in capturing the dynamics of real markets have been well documented. When using the Black-Scholes model, the implied volatilities derived from market data can be used to construct more accurate hedging portfolios. Moreover, jump models can be employed to modify the Black-Scholes model in order to better represent real-world financial crises. By incorporating jumps, a theoretical study can be made more realistic and consistent with real-world scenarios. This will allow to analyse the impact of sudden losses on the performance of EPSs and the evaluation of the efficiency of different hedging strategies in protecting investors against dramatic losses. By providing valuable insights for both providers and holders, a theoretical study of a model with jumps would help to improve the design and usage of EPSs.

Finally, an important direction for future research is the case of a cross-currency reference portfolio. In the context of Australian super funds, it is fairly common for members to hold assets across several economies. Hence to generalise the application scope of an EPS, the pricing and hedging strategies should take into account the exchange rate fluctuations and the correlation between prices of domestic and foreign assets. Various existing approaches to the arbitrage-free pricing and hedging of cross-currency reference portfolios include to either consider two independent local returns on domestic and foreign assets or focus on the total return on the portfolio's value expressed in the domestic currency. This area of research is essential for providers and investors operating in international markets, as it will provide a valuable insight into fair pricing and hedging strategies for an EPS in a global context.

\vskip 10 pt \noindent  {\large \bf Acknowledgments} \vskip 5 pt
The research of H. Xu, R. Liu, and M. Rutkowski was supported by the Australian Research Council Discovery Project scheme
under grant DP200101550 {\it Fair pricing of superannuation guaranteed benefits with downturn risk}.

\end{document}